\documentclass[final,leqno,onefignum,onetabnum]{siamltex1213}
\usepackage{amsmath,amssymb,amstext}
\usepackage[figure,linesnumbered,vlined]{algorithm2e}
\usepackage{multirow}
\usepackage{arydshln}
\usepackage{enumitem,framed}

\newcommand{\R}{\mathbb{R}}
\newcommand{\N}{\mathbb{N}}

\newcommand{\B}{\mathbb{B}}

\newcommand{\I}{\mathcal{I}}
\newcommand{\calL}{\mathcal{L}}
\newcommand{\rk}{\text{rk}}

\DeclareMathOperator{\support}{supp}

\renewcommand{\mid}{:}

\renewcommand{\P}{\mathbb{P}}
\renewcommand{\vec}{\mathbf}
\newcommand{\norm}[1]{\lVert #1 \rVert_1}

\newcommand{\dmax}{\delta}

\newtheorem{example}{Example}

\newtheorem{assumption}{Assumption}

\newtheorem{observation}{Observation}

\title{Sensitivity Analysis for Convex Separable Optimization over Integral Polymatroids}

\author{
Tobias Harks\thanks{Institute for Mathematics, University Augsburg, Germany.\newline
(\email{tobias.harks@math.uni-augsburg.de}).}
\and
Max Klimm\thanks{School of Business and Economics, Humboldt University Berlin, Germany.\newline
(\email{max.klimm@hu-berlin.de}).}
\and
Britta Peis\thanks{School of Business and Economics, RWTH Aachen University, Germany.\newline
(\email{britta.peis@oms.rwth-aachen.de}).}
}

\begin{document}
\maketitle
\slugger{sidma}{xxxx}{xx}{x}{x--x}

\begin{abstract}
We study the sensitivity of optimal solutions of convex separable optimization problems over an integral polymatroid base polytope with respect to parameters determining both the cost of each element and the polytope. Under convexity and a regularity assumption 
on the functional dependency of the cost function with respect to the parameters,
we show that reoptimization after a change in parameters can be done by elementary local operations. Applying this result, we derive that starting from any optimal solution there is a new optimal solution to new parameters such that the $L_1$-norm of the difference of the two solutions is at most two times the $L_1$-norm of the difference of the parameters.

We apply these sensitivity results to a class of non-cooperative games with a finite set of players where a strategy of a player is to
choose a vector in a player-specific integral polymatroid base polytope defined on a common set of elements. The players' 
private cost functions are regular, convex-separable and the cost of each element is a non-decreasing function of the own usage of that element and the overall usage of the other players. Under these assumptions,
we establish the existence of a pure Nash equilibrium. The existence is proven by an algorithm computing a pure Nash equilibrium  that runs in polynomial time 
whenever the rank of the polymatroid base-polytope is polynomially bounded. 
Both the existence result and the algorithm generalize and unify previous results appearing in the literature.

We finally complement our results by showing that polymatroids are the maximal combinatorial structure enabling these results. For \emph{any} non-polymatroid region, there is a corresponding optimization problem for which the sensitivity results do not hold. In addition, there is a game where the players' strategies are isomorphic to the non-polymatroid region and that does not admit a pure Nash equilibrium.
 
\end{abstract}



\pagestyle{myheadings}
\thispagestyle{plain}
\markboth{HARKS, KLIMM, AND PEIS}{SENSITIVITY ANALYSIS FOR 
INTEGRAL POLYMATROID OPTIMIZATION}

\section{Introduction}
We consider polymatroid optimization problems, where the objective is to distribute $d\in\N$ discrete units among a set of elements $E=\{1,\dots,m\}$ so as to minimize a convex separable
cost function subject to upper bounds on the total amount of units
allocated to subsets of elements. These upper bounds are defined via values of an integral polymatroid rank function $f:2^E\rightarrow \N$.
Formally, we study the following optimization problem:
\begin{equation}
\tag{$P(\vec t, d)$}\label{problem}
\begin{aligned}
 \text{minimize } & \sum_{e\in E} 
C_e(x_e;t_e)\\ 
\text{subject to: } & \sum_{e\in U}x_e \leq f(U)\text{ for all }U\subseteq E \\
& \sum_{e \in E}x_e = d\\
& x_e\in \N \text{ for all }  e \in E,
\end{aligned}
\end{equation}
where the functions $C_e: \N\times \N \to \R_+, e \in E$ are non-decreasing and discrete convex in the first entry and $f$ is a normalized,  monotone and submodular set function. 
The vectors $\vec t=(t_e)_{e \in E}\in\N^{|E|}$ and $d\in\N$ are integral parameters.
For fixed parameters, this problem is a convex-separable optimization problem over an integral polymatroid base polytope and can be solved in polynomial time by greedy algorithms; see Federgruen and Groenevelt~\cite{FedergruenG86b}, Groenevelt~\cite{groenevelt91}, Hochbaum and Shanthikumar~\cite{HochbaumS90} and the book by Fujishige~\cite{fujishige2005submodular}. 
Besides these appealing theoretical properties, the problem has applications in several areas ranging from 
scheduling problems (cf.~Yao~\cite{Yao02} and Krysta et al.~\cite{KrystaSV03}), 
and tree packing and matroid optimization (cf.~Gabow~\cite{Gabow95})
to game-theoretic applications (cf.~He et al.~\cite{HeZZ12}). 

\subsection{Sensitivity Analysis}
\label{subsec:sensitivity_analysis}
Suppose we are given an optimal solution $\vec x^*(\vec t, d)$ with respect
to  $\vec t$ and  $d$. The main question addressed in this paper is: How does the structure of  an optimal solution change after changes to the parameters $\vec t$ and $d$? We motivate this question by a concrete example.
Let  $G=(V,E)$ be a connected undirected graph
with vertex set $V$ and edge set $E\subseteq (V\times V)$.
The objective is to compute $k\in \N$ spanning trees of $G$ so that along each spanning tree a message of unit size can be sent.  If $x_e\in \N$ messages are sent along edge $e$, i.e., $e$ is contained in exactly $x_e$ spanning trees, the resulting (average) delay is defined as \[ C_e(x_e;u_e)=\begin{cases}\frac{1}{u_e-x_e}, &\text{ if } x_e<u_e,\\
+\infty, &\text{ else,}
\end{cases}\]
where $C_e(x_e;u_e)$ is a standard $M/M/1$-delay function frequently used in queueing theory \cite{GaiLK11,KorilisL95}. The parameter $u_e\in \N$ denotes the installed capacity on edge $e$. 
The problem to compute $k$ spanning trees to minimize the total delay can be cast as a convex separable integral polymatroid optimization problem by taking $f$ as the $k$-th multiple of the rank function of the graphic matroid on $G$. In this model,
it is natural to ask how optimal solutions change if the edge capacities or $k$ are changed.

\subsection{Our Results for the Sensitivity of Polymatroid Optimization}
The change of an optimal solution for changed parameters clearly depends on the structural dependency of the objective function and the feasible region on the parameters $\vec t$ and  $d$.
To capture this dependency, we introduce the following concept of \emph{regularity}.
Informally, we call a function $C(x;t)$ regular, if (i)
the (left discrete) partial derivative with respect to the first entry is a nondecreasing function in the parameter; (ii) the left discrete partial derivative with respect to the first entry
is not larger after a unit increase of the parameter $t$ than after a unit increase of $x$.

Our main results (Theorem~\ref{thm:shift}, Theorem~\ref{thm:dem} and Theorem~\ref{thm:gen}) can be informally summarized as follows:
Let $\vec x^*(\vec t, d)$ be an optimal solution of $P(\vec t, d)$
for regular and convex functions $C_e, e \in E$.
Then, for any other integral parameters $\vec t', d'$, there exists a new optimal solution $\vec x^*(\vec t', d')$
\emph{close} to $\vec x^*(\vec t, d)$ in the following sense:
\begin{align*}\norm{\vec x^*(\vec t,d)-\vec x^*(\vec t',d')}\leq  2\norm{\vec t - \vec t'}+|d-d'|.
\end{align*}
Moreover, given  $\vec x^*(\vec t, d)$, we can compute $\vec x^*(\vec t', d')$
by performing $\norm{\vec t - \vec t'}+|d-d'|$ elementary exchange steps.

In the context of the tree packing example with $M/M/1$ queueing functions discussed in \S~\ref{subsec:sensitivity_analysis}, this sensitivity result implies that if the capacity vector is changed,  say by $p$ units, then
there is a new optimal solution  for which at most $2p$ edges of the given spanning trees needs to be changed, and these changes can be efficiently computed from the initial optimal solution. On the other hand, when $k$ is increased by $p$ units, then at most $p|V|$ units need to be changed since adding a single additional spanning tree corresponds to adding $|V|$ new units.

In \S~\ref{sec:violation}, we show that submodularity of the capacity function defining the base polytope is necessary for the sensitivity results above in the following sense: For any capacity function that is not submodular, both sensitivity results (in terms of a change of $\vec t$ and $d$) do not hold anymore. 

\subsection{Our Results for Polymatroid Games}
Our sensitivity results have consequences for the existence
of pure Nash equilibria for a new class of non-cooperative
games that we term \emph{polymatroid games}.
In such a game, there is a finite set of players $N$ and
a finite set of resources $E$. Each player $i\in N$ distributes
her demand $d_i\in \N$ in integral units among the resources subject to player-specific submodular capacity constraints. This way, the resulting strategy space for each player forms an integral polymatroid base polytope (truncated at the player-specific demand).
We further assume that the private cost function is defined
as a sum of regular, player-specific, non-decreasing, convex cost functions $C_{i,e}(x_e;t_e)$, where the parameter $t_e$
is interpreted as the load contribution of other players to resource $e$.
This class of games includes as special cases the following variants of
congestion games that have been treated separately in the literature.

\begin{enumerate}[label=(\roman*),]
\item congestion games with matroidal strategies and player-specific costs (studied by Ackermann et al.~\cite{AckermannRV09}) for the case that the submodular capacity function is the rank function of a matroid;
\item integer-splittable congestion games with singleton strategies (studied by Tran-Thanh et al.~\cite{TranPCRJ11}) for the case that the submodular capacity function is equal to a sufficiently large constant;\label{it:singleton-integer-splittable}
\item integer-splittable congestion games with matroidal strategies (the matroidal counterpart of integer-splittable congestion games studied by Rosenthal~\cite{Rosenthal73integers}) for the case that the submodular capacity function is an integer multiple of the rank function of a matroid.\label{it:matroid-integer-splittable}
\end{enumerate}
We show the existence of a pure Nash equilibrium and devise an algorithm for its computation for polymatroid games thus generalizing and unifying the existence results by Ackermann et al.\ and Tran-Thanh et al. Our algorithm maintains \emph{preliminary demands, strategy spaces, and strategies}
of the players that all are set to zero initially.
In the course of the algorithm the demands of the players are increased iteratively by one unit
and a preliminary pure Nash equilibrium (with respect
to the current demands) is recomputed by following a sequence of best response moves
of the players. While similar algorithms have been proposed before for matroid congestion games (cf. Ackermann et al.~\cite{AckermannRV09}) and integer-splittable singleton games (cf. Tran-Thanh et al.~\cite{TranPCRJ11}), the main difficulty is to show that the sequence of best responses converges even in the more general setting of polymatroid games. This is where the sensitivity results for polymatroid optimization shown in \S~\ref{sec:sensitivity} are applied.
Furthermore, we show that the runtime of the algorithm is bounded by $n^2 m\dmax^{3}$,
where $n$ is the number of players, $m$ the number of resources, and
$\dmax$ is an upper bound on the maximum demand. Thus, for polynomially bounded
$\dmax$, the algorithm is polynomial.

In \S~\ref{sec:violation}, we prove that  submodularity of the players' capacity constraints is necessary for the existence of a pure Nash equilibrium in a strong sense. For \emph{any} monotonic, normalized, strictly positive, and non-submodular set function defining the player-specific base polytope, there is a two-player game without a pure Nash equilibrium. In this sense, our results are best possible and polymatroids are the maximal combinatorial structure guaranteeing the existence of a pure Nash equilibrium.

\subsection{Related Work}
\paragraph{Sensitivity in Polymatroid Optimization}

Fujishige et al.~\cite{FujishigeGHP15} studied
convex-separable minimization problems over polymatroid
base polytopes. They conduct a sensitivity analysis of optimal solutions
for the case that the marginal cost
function is  (component-wise) shifted downwards. As their main result,
they show that for any downwards shift of the marginal cost function, any new optimal
solution has the property that all optimal  cost values of the components
decrease monotonically with respect to the original optimal solution.
Their motivation is to analyze the Braess paradox, and their sensitivity result
implies that there is no Braess paradox in polymatroids.
A difference to our work lies in the fact that we conduct a \emph{quantitative} sensitivity analysis, that is, given an optimal solution for an initial parameter, we quantify exactly the difference of a closest new optimal solution for a changed parameter.

He et al.~\cite{HeZZ12} considered separable-concave maximization problems
over polymatroids in which each cost function component has a second parameter
and is concave in the parameter. They assume that the cost functions are discrete convex in both entries. A further difference to our setting is that we consider as
feasible domain an \emph{integral}
polymatroid base polytope instead of an ordinary polymatroid (defined by a real-valued
submodular function). 
Their main result establishes  (without any further regularity assumptions) that the optimal value function
 as a function of the parameter, or as a function of the support set of the objective function, is submodular. 
 These two result have important consequences 
 for game-theoretical models, because, using the submodularity of the optimal value function,
 they show that the joint replenishment game and the one-warehouse multiple retailer game is submodular and, thus, has
  desirable properties in terms of existence of core solutions.
  In contrast to the work of He et al.~\cite{HeZZ12}, we study in this paper
  sensitivity properties of underlying optimal solutions and not the optimal value function.
  Moreover, the structure of an integral polymatroid base polytope differs
  from an ordinary polymatroid. The integral polymatroid base polytope does not
 form a lattice when considering the component-wise minimum and maximum
 as join and meet. Thus, the sensitivity framework of Topkis~\cite{Topkis98,Topkis78}
 (to which also He et al.~\cite{HeZZ12} refer)  is not directly applicable.

 Cook et al.~\cite{CookGST86} considered general integer programs of
the form $\max\{w x\;\vert\; Ax\leq b\}$ and conducted a proximity and sensitivity
analysis. The proximity analysis is concerned with the difference of optimal solutions for integer linear programs and their continuous relaxations, respectively (the difference is measured by the $L_{\infty}-$ or $L_1$-norm).
The sensitivity analysis investigates the difference of optimal solutions of an integer linear program for changed $b$.  Their main result shows that the $L_{\infty}-$ distance to the nearest optimal solution to the corresponding integer program is at most the number of variables multiplied by the largest sub-determinant of $A$. Baldick~\cite{BALDICK1995}
strengthened some of the sensitivity and proximity results of Cook et al.~\cite{CookGST86} by introducing a finer measure for the constraint matrix $A$.

Murota~\cite{MurotaT04} and later Moriguchi et al.~\cite{MoriguchiST11}
derived  proximity results for the minimization of M-convex functions in integer variables
(see Murota~\cite{Murota:2003} for an introduction  to this concept).
Moriguchi et al.~\cite{MoriguchiST11} showed
for convex-separable objective that the 
difference between integral and fractional optimal solutions measured in the $L_1$-norm
is at most $2(n-1)$, where $n$ is the number of variables.

\paragraph{Congestion Games}
Rosenthal~\cite{Rosenthal73congestion} introduced congestion
games, a class of strategic games, where
a finite set of players competes over a finite set of
resources.  Each player is associated with a set of allowable subsets of resources, e.g., implicitly given by a combinatorial property. A pure strategy of a player is to choose a subset from this set. 
In the context of \emph{network games}, the resources correspond to the edges of a graph and the
allowable subsets correspond to the paths connecting a given source and a given sink.
The cost of a resource depends only on the number of
players choosing the same resource and each player strives to minimize
the sum of the costs of the resources contained in the selected subset.
In this general model, Rosenthal proved the existence of a pure Nash equilibrium by a potential function argument. Up to date congestion games
have been used as reference models for decentralized systems involving the selfish allocation of congestible resources (e.g., selfish
route choices in traffic networks~\cite{BeckmannMW56,Roughgarden05book,Wardrop52}
and flow control in telecommunication networks~\cite{JohariT06,KellyMT98,Srikant03})
and for decades they have been a focal point of research
in (algorithmic) game theory, operations research and theoretical computer science. 
 
In the past, the existence of pure Nash equilibria has been analyzed in many variants
of congestion games such as singleton congestion games with player-specific cost
functions (cf.~\cite{GairingMT11,IeongMNSS05,Milchtaich96,Milchtaich06}), congestion games with weighted players 
(cf.~\cite{AckermannRV09,AnshelevichDKTWR08,ChenR09,FotakisKS05,HarksK12}), nonatomic and atomic splittable congestion games 
(cf.~\cite{BeckmannMW56,HaurieM85,KellyMT98,Wardrop52})
and congestion games with player- and resource-specific and variable demands (cf.~\cite{HarksK16}).

Rosenthal proposed congestion games with \emph{integer-splittable demands} as an important
and meaningful model already back in 1973 in his first work on congestion games \cite{Rosenthal73integers}
even published prior to his more famous work~\cite{Rosenthal73congestion}. 
  Despite this long history, not much is known
regarding existence and computability of pure Nash equilibria.
Rosenthal gives an example that shows that pure Nash equilibria
need not exist for integer-splittable congestion games in general. Dunkel and Schulz~\cite{DunkelS08} strengthened this result showing that the existence of a pure Nash equilibrium in integer-splittable congestion games is NP-complete to decide. Meyers~\cite{Meyers08} proved that in games with linear cost functions, a pure Nash equilibrium is always guaranteed to exist. For singleton strategy spaces and nonnegative and convex cost functions, Tran-Thanh et al.~\cite{TranPCRJ11} showed the existence of pure Nash equilibria. They also showed that
pure Nash equilibria need not exist (even for the restricted strategy spaces)
if cost functions are semi-convex. Our results generalize the existence result of Tran-Thanh et al.\ towards general polymatroid strategy spaces.

Parts of the results of this paper have been presented by the authors in less general and preliminary form in the Proceedings of the 10th Conference on Web and Internet Economics~\cite{HarksKP14}.

\section{Preliminaries}
Let $\N$ denote the set of nonnegative integers and let $E$ be a finite and non-empty set of elements. We write $\N^E$ shorthand for $\N^{|E|}$. Throughout this paper, vectors $\vec x = (x_e)_{e \in E} \in \N^E$ will be denoted with bold face. An integral set function $f : 2^E \rightarrow \N$ is \emph{submodular} if $f(U) + f(V) \geq f(U \cup V) + f(U \cap V)$ for all $U,V \in 2^E$; $f$ is \emph{monotone} if
$f(U)\le f(V)$ for all $U\subseteq V \subseteq E$; and $f$ is \emph{normalized} if $f(\emptyset)=0$.
We call an integral, submodular, monotone, and normalized function $f : 2^E \rightarrow \N$ an \emph{integral polymatroid rank function}. The associated \emph{integral polyhedron} is defined as
\begin{align*}
\P_f &:= \Bigl\{\vec x \in \N^{E} \mid x(U) \leq f(U)\text{ for all } U\subseteq E \Bigr\}, 
\intertext{where for a vector $\vec x = (x_e)_{e \in E}$ and $U \subseteq E$, we write $x(U)$ shorthand for $\sum_{r \in U} x_e$. For an element $e \in E$, we write $x(e)$ instead of $x(\{e\})$. Given the integral polyhedron $\P_f$ and an integer $d \in \N$ with $d \leq f(E)$, the \emph{$d$-truncated integral polymatroid} $\P_f(d)$ is  defined as}
\P_f(d) &:= \Bigl\{\vec x\in \N^{E} \mid x(U) \leq f(U)\text{ for all }U\subseteq E,\, x(E) \leq d \Bigr\}.
\intertext{The $d$-truncated integral polymatroid $\P_f(d)$ is again an integral polyhedron as the corresponding integral polymatroid rank function $f' : 2^E \to \N$ defined as $f'(U) = \min\{d,f(U)\}$ is submodular. For a $d$-truncated integral polymatroid $\P_f(d)$, the corresponding \emph{integral polymatroid base polyhedron} is defined as}
\B_f(d) &:= \Bigl\{\vec x\in \N^{E} \mid x(U) \leq f(U)\text{ for all }U\subseteq E, x(E) = d \Bigr\}.
\end{align*}
The \emph{rank} of $\B_f(d)$ is given by $d$.
We consider the problem of minimizing a 
parametrized separable discrete convex function over a polymatroid base polytope. 
\begin{framed}
\begin{align}\label{opt-problem}  \text{minimize } &\sum_{e \in E} 
\tag{$P(\vec t,d)$} C_e(x_e;t_e)\\ \notag
\text{subject to: } &\vec x\in \B_f(d),
\end{align}
where $\vec t= (t_e)_{e \in E}\in\N^{E} , d\in\N$ are nonnegative integral parameters.
\end{framed}

Note that $\vec t$ influences the cost function while the parameter $d$
defines the truncation of the integral polymatroid base polytope. Let $\vec x^*(\vec t,d)\in \N^{E}$ be an optimal solution to \ref{opt-problem}. We are interested in quantifying the distance between $\vec x^*(\vec t,d)$ and a new optimal solution $\vec x^*(\vec t',d')$ for the new parameters $\vec t'\in \N$ and $d'\in \N$. We will measure the distance between solutions using the $L_1$-norm defined as $\norm{\vec x}:=\sum_{e \in E}|x_{e}|$ for all $\vec x \in \N^E$.
Thus, we are interested in bounding $\norm{\vec x^*(\vec t,d)-\vec x^*(\vec t',d')}$ in terms of  $\norm{\vec t-\vec t'}$ and $|d-d'|$. Naturally, non-trivial bounds are only possible when imposing additional assumptions on the dependence of the cost function on the parameter.

To state this assumption formally, we need the following notation of discrete derivatives. For a function $c : \N \to \R_+$ and $x \in \N$, let
\begin{align*}
c^-(x)&:= c(x)-c(x-1) & &\text{and} & c^+(x)&:=c(x+1)-c(x)
\intertext{denote the left and right discrete derivative, respectively. Note that the left derivative is only defined for $x \geq 1$. 
The function $c$ is called \emph{discrete convex,} if $c^-(x) \leq c^+(x)$ for all $x \in \N$.
For a function $C:\N\times\N\rightarrow \R_+$, we slightly abuse notation as we denote by}
C^-(x;t) &:=C(x;t)-C(x-1;t) & &\text{and} & C^+(x;t)&:=C(x+1;t)-C(x;t),
\end{align*}
the left and right derivatives, respectively, with respect to the first argument. 
We call $C$ discrete convex, if $x\mapsto C(x;t)$ is discrete convex
for all $t\in \N$.


We introduce the following notion of regularity that bounds the impact of a parameter change on the derivative.
\begin{definition}[Regularity]
A function $C: \N \times \N \to \R$ is called \emph{regular},
if \begin{align}
\label{reg:mon}C^-(x;t)&\leq  C^-(x;t+1) \text{ for all $x,t\in \N$,}\\
\label{reg:shift} C^-(x;t+1)&\leq  C^-(x+1;t) \text{ for all $x,t\in \N$.}
\end{align}
\end{definition}
In words,~\eqref{reg:mon}  requires that 
the (left) marginal cost function of $C$ is nondecreasing in $t$ and~\eqref{reg:shift} 
bounds the marginal cost of $C$
after adding one unit to parameter $t$
in terms of the marginal cost after adding one unit to $x$.
It can be shown that a regular function $C$ is discrete convex.
\begin{observation}
A regular function $C: \N \times \N \to \R$ is discrete convex.
\end{observation}
\begin{proof}
We calculate
\[ C^-(x;t)\leq  C^-(x;t+1) \leq  C^-(x+1;t)=C^+(x;t),\]
where for the first and second inequality we used \eqref{reg:mon} and \eqref{reg:shift}, respectively.
\end{proof}

Throughout this work, we impose the following assumption on $C_e, e \in E$.
\begin{assumption}\label{ass:cost-general}
For every $e \in E$, $C_e$ is regular.
\end{assumption}

If \ref{opt-problem} only involves cost functions satisfying Assumption~\ref{ass:cost-general}, 
we speak of a convex and regular optimization problem.

We recapitulate the motivating example given in \S~\ref{subsec:sensitivity_analysis}
and state it as a convex and regular optimization problem.
\begin{example}
Let  $G=(V,E)$ be a connected undirected graph.
The objective is to compute $k\in \N$ spanning trees of $G$ with minimum cost so that along each spanning tree a message of unit size can be sent.  If $x_e\in \N$ messages are sent along edge $e$, the resulting cost per edge is defined as \[ C_e(x_e;u_e)=\begin{cases}\frac{1}{u_e-x_e}, &\text{ if } x_e<u_e,\\
+\infty, &\text{ else.}
\end{cases}\]
By defining $u = \max_{e \in E} u_e$ and $t_e=u-t_e$ for all $e\in E$, we
obtain an equivalent problem in which  the cost functions are of the form $C_e(x_e;t_e) = 1/(u-t_e-x_e)$. This problem involves regular cost functions, because  $\partial C_e(x_e;t_e)/\partial x_e = 1/(u-t_e-x_e)^2$ is increasing in $t_e$, thus,~\eqref{reg:mon} is satisfied. Moreover, one easiliy
verifies that also~\eqref{reg:shift} is satisfied.
\end{example}

\section{Sensitivity Results}
\label{sec:sensitivity}
For fixed parameters $\vec t \in \N^E, d \in \N$, we recapitulate the following necessary and sufficient optimality conditions for an optimal solution to problem~\ref{opt-problem}.
Let $\chi_e\in\N^{E}$ be the indicator vector with all-zero
entries except for the $e$-th coordinate which is $1$.
For $\vec x\in \B_f(d)$ and $e \in E$ denote by 
\begin{align*}
D_e(\vec x) &=\{g \in E\setminus\{e\}\;|\; \vec x+ \chi_{g}-\chi_e  \in \B_f(d)\}
\intertext{the set of \emph{feasible local exchanges} w.r.t.\  $\vec x$ and $e$ and by}
\Delta_e(\vec x; \vec t) &=
\begin{cases}
\min_{g \in D_e(\vec x)} C^+_g(x_g;t_g), & \text{ if $D_e(\vec x) \neq \emptyset$,}\\
+\infty, &\text{ else,}
\end{cases}
\end{align*}
the minimum alternative cost when exchanging $e$. The following theorem gives a necessary and sufficient condition for the optimality of a solution of a polymatroid optimization problem.

\begin{theorem}[Fujishige~\cite{fujishige2005submodular}]\label{thm:fuj}
$\vec x^*\in \B_f(d)$ is an optimal solution for~\ref{opt-problem} if and only if $C^-_e(x_e^*;t_e)\leq \Delta_e(\vec x^*; \vec t)$ for all $e \in E$.
\end{theorem}

Using these conditions, we
proceed to establish the first sensitivity result which relates optimal
solutions for changed values of $\vec t$. 

\begin{theorem}\label{thm:shift}
Let $P(\vec t, d)$ be a regular convex optimization problem with optimal solution
$\vec x^*(\vec t,d)$ and let 
$\vec t'=\vec t + \chi_{e^*}$ for some $e^* \in E$. Let
\begin{align*}
g^*\in
\begin{cases}
	\arg\min\nolimits_{g \in D_{e^*}(\vec x^*(\vec t,d))}  \{C_g^+(x_g; t_g)\}, &\text{if $D_{e^*}(\vec x^*(\vec t,d)) \neq \emptyset$,}\\
	\{e^*\},  &\text{else.}
\end{cases}
\end{align*}
Then the better of the two solutions,
$\vec x^*(\vec t, d)$ and $\vec x^*(\vec t, d) -\chi_{e^*} +\chi_{g^*}$, 
is optimal for $P(\vec t', d)$.
\end{theorem}

\begin{proof}
When $D_{e^*}(\vec x^*(\vec t, d)) = \emptyset$, then $\vec x^*(\vec t, d)$ trivially satisfies the optimality conditions of Theorem~\ref{thm:fuj} for any parameter vector $\vec t'$ and there is nothing left to show. Thus, let us assume $D_{e^*}(\vec x) \neq \emptyset$. 
Let $\vec x=\vec x^*(\vec t,d)$ and $\vec y = \vec x - \chi_{e^*} + \chi_{g^*} \neq \vec x$.


First, consider the case $C_{e^*}^-(x_{e^*}; t_{e^*}+1)\leq \Delta_{e^*}(\vec x; \vec t)$. We claim that then $\vec x$ is still an optimal solution to $P(\vec t',d)$ as it
satisfies the optimality conditions
\begin{align}
\label{eq:optimality_thm32}
C_e^-(x_e; t_e') \le \Delta_e(\vec x; \vec t') \quad \text{ for all $e \in E$.}
\end{align}
To see this,  note that by Assumption~\ref{ass:cost-general} (using~\eqref{reg:mon}) we get $\Delta_e(\vec x; \vec t') \geq \Delta_e(\vec x; \vec t)$ for all $e \in E$ as $t_e' \geq t_e$. This directly implies that \eqref{eq:optimality_thm32} is satisfied for all $e \neq e^*$. To see the inequality also for $e^*$, observe that
\begin{align*}
C_{e^*}^-(x_{e^*};t_{e^*}') = C_{e^*}^-(x_{e^*}; t_{e^*}+1)\leq \Delta_{e^*}(\vec x; \vec t) =\Delta_{e^*}(\vec x; \vec t'),
\end{align*}
where for the last equality we used that $C_{g}^+(x_g;t_g) = C_{g}^+(x_g;t_g')$ for all $g \neq e^*$.

Second, consider the case $C^-_{e^*}(x_{e^*};t_{e^*}+1) > \Delta_{e^*}(\vec x; \vec t)$.  
We proceed to show that then $\vec y$ is optimal for $P(\vec t', d)$ by checking the optimality conditions of Theorem~\ref{thm:fuj} for all $e \in E=\{e^*\}\cup \{g^*\}\cup E\setminus\{e^*,g^*\}$.

{\bf Case ({\bf{A}}):} $e=e^*$. We obtain
\[ C_{e^*}^-(y_{e^*}; t_{e^*}+1)=C_{e^*}^-(x_{e^*}-1; t_{e^*}+1)\leq C_{e^*}^-(x_{e^*}; t_{e^*})\leq \Delta_{e^*}(\vec x; \vec t),\]
where the first inequality follows by the regularity of $C_{e^*}$ and the second since $\vec x$
is optimal for $P(\vec t, d)$. Thus, the optimality conditions for $e^*$ are satisfied if $\Delta_{e^*}(\vec x; \vec t) \leq \Delta_{e^*}(\vec y; \vec t')$. To prove the latter inequality, we first show that $D_{e^*}(\vec y) \subseteq D_{e^*}(\vec x)$. Assume by contradiction that there is  $g\in D_{e^*}(\vec y)\setminus D_{e^*}(\vec x)$.
This implies 
$ \vec x -\chi_{e^*}+\chi_{g}\notin \B_f(d).$
Hence, there must exist $T\subset E$ with $g\in T, e^*\notin T$
and $f(T)=x(T)$. On the other hand, $g\in D_{e^*}(\vec y)$ implies
$\vec y':=\vec y-\chi_{e^*}+\chi_{g}= \vec x -2\chi_{e^*}+\chi_{g}+\chi_{g^*}\in \B_f(d)$.
Using $g\in T, e^*\notin T, x(T)=f(T)$, this implies $y'(T)\geq x(T)+1= f(T)+1$, hence, $\vec y'\notin \B_f(d)$,
a contradiction.

Finally, let us show $\Delta_{e^*}(\vec x; \vec t) \leq \Delta_{e^*}(\vec y; \vec t')$. This is trivial if $D_{e^*}(\vec y) = \emptyset$. Otherwise, we obtain 
\begin{align*}
\Delta_{e^*}(\vec x; \vec t) &= \min_{g \in D_{e^*}(\vec x)} \{C^+_g(x_g; t_g)\} = C^+_{g^*}(x_{g^*};t_{g^*}) \leq C^+_{g^*}(y_{g^*};t_{g^*}) = C^+_{g^*}(y_{g^*};t'_{g^*})
\intertext{where for the inequality we used discrete convexity. Moreover, we obtain for all $g' \in D_{e^*}(\vec x) \setminus \{g^*\}$ the inequality}
\Delta_{e^*}(\vec x; \vec t) &= \min_{g \in D_{e^*}(\vec x)} \{C^+_g(x_{g}; t_{g})\} \leq C^+_{g'}(x_{g'}; t_{g'}) = C^+_{g'}(y_{g'};t'_{g'}).
\end{align*}
Putting things together, we get
\[\Delta_{e^*}(\vec x; \vec t)\leq \min_{g \in D_{e^*}(\vec x)} C^+_g(y_g; t'_g)\leq  \min_{g \in D_{e^*}(\vec y)} C^+_g(y_g; t'_g) \leq  \Delta_{e^*}(\vec y; \vec t') .\]

{\bf Case ({\bf{B}}):} $e=g^*$.
For a contradiction, assume that there is  $g\in D_{g^*}(\vec y)$
with 
\begin{equation}\label{contra1}C_{g^*}^-(y_{g^*}; t_{g^*}')>C_{g}^{+}(y_{g}; t_{g}').
\end{equation}
Note that $g\neq e^*$ since
\begin{multline*}
	C_{g^*}^-(y_{g^*}; t'_{g^*}) = C^-_{g^*}(y_{g^*};t_{g^*}) = C_{g^*}^+(x_{g^*}; t_{g^*})\\
<C_{e^*}^{-}(x_{e^*}; t_{e^*}+1)=C_{e^*}^{+}(y_{e^*}; t_{e^*}+1) = C_{e^*}^+(y_{e^*};t_{e^*}'),
\end{multline*}
a contradiction to~\eqref{contra1}.
Thus,  
$g\in D_{g^*}(\vec y) \setminus \{e^*\}$. We obtain
\[ \vec y-\chi_{g^*}+\chi_{g}= \vec x -\chi_{e^*}+\chi_{g} \in \B_f(d)\]
which implies $g\in D_{e^*}(\vec x)$.
Since $g^*$  minimizes $C_{g}^{+}(\vec x; \vec t)$ among all $g\in D_{e^*}(\vec x)$, we get
$C_{g^*}^-(y_{g^*}; t'_{g^*}) = C_{g^*}^-(y_{g^*}; t_{g^*})=C_{g^*}^+(x_{g^*}; t_{g^*})\leq C_{g}^{+}(x_{g}; t_{g}) = C_{g}^{+}(y_{g}; t_{g}')$
contradicting~\eqref{contra1}.

{\bf Case ({\bf{C}}):} $e \in E\setminus \{e^*,g^*\}$.
Assume by contradiction that there is $g \in D_{e}(\vec y)$
with 
\begin{equation}\label{eq:con} C_{e}^-(y_{e}; t'_{e})>C_{g}^{+}(y_{g}; t'_{g}).
\end{equation}
We first treat the case $g=e^*$ where \eqref{eq:con} becomes
\begin{equation}\label{contra2}
C_{e}^-(y_{e}; t_{e}) > C_{e^*}^{+}(y_{e^*}; t_{e^*}+1).
\end{equation}
With $e^*\in D_{e}(\vec y)$, we get $\vec y-\chi_e+\chi_{e^*}= \vec x -\chi_e+\chi_{g^*}\in \B_f(d)$,
hence, $g^*\in D_e(\vec x)$.
We get
\begin{align*} C_{e^*}^{+}(y_{e^*}; t_{e^*}+1)&<C_{e}^-(y_{e}; t_{e}) \tag{using \eqref{contra2}}\\
&=C_{e}^-(x_{e}; t_{e}) \tag{using  $e \in E\setminus \{e^*,g^*\}$} \\
& \leq C_{g^*}^{+}(x_{g^*}; t_{g^*}) \tag{since $\vec x$ was optimal for $\vec t$}\\
&<C_{e^*}^{-}(x_{e^*}; t_{e^*}+1)\tag{since $\vec x$ was not optimal for $\vec t'$} \\
&=C_{e^*}^{+}(y_{e^*}; t_{e^*}+1),\tag{since $\vec y=\vec x-\chi_{e^*}+\chi_{g^*}$}
\end{align*}
a contradiction.

From now on we may assume $g \in E \setminus \{e^*\}$.
We claim that the following two properties are satisfied:
\begin{enumerate}[label=(\alph*)]
\item\label{it:condition_a_in32} $\vec x -\chi_{e}+\chi_{g^*} \in \B_f(d),$
\item\label{it:condition_b_in32} $\vec x -\chi_{e^*}+\chi_{g} \in \B_f(d)$.
\end{enumerate}
This claim has also been used in the proof of Theorem 4.11 in Fujishige~\cite{fujishige2005submodular}.
In order to keep our analysis self-contained, we provide an alternative proof of fact \ref{it:condition_a_in32} and~\ref{it:condition_b_in32} below.

Before proving these properties, however, we show that they give the desired contradiction to \eqref{eq:con}. First note that \ref{it:condition_a_in32} implies $g^* \in D_{e}(\vec x)$ which together with the optimality of $\vec x$ for $\vec t$ implies $C_e^-(y_e;t_e') = C^-_e(x_e;t_e) \leq C_{g^*}^+(x_{g^*}; t_{g^*})$. Similarly,  \ref{it:condition_b_in32} implies $g \in D_{e^*}(\vec x)$ which, by the choice of $g^*$, implies $C^+_{g^*}(x_{g^*};t_{g^*}) \leq C^+_{g}(x_g;t_g)$. Finally, we have $C^+_{g}(x_g;t_g) \leq C^+_{g}(y_{g}; t_{g}) = C^+_{g}(y_g; t'_g)$ by discrete convexity and the fact that $y_g \geq x_g$ for all $g \in E \setminus \{e^*\}$. Combining all three inequalities, we obtain $C_e^-(y_e;t_e') \leq C_g^+(y_g;t_g')$, a contradiction to \eqref{eq:con}.

It remains to prove properties \ref{it:condition_a_in32} and \ref{it:condition_b_in32}.
Denote
\[ \vec y':=\vec y-\chi_e +\chi_g= \vec x-\chi_{g^*}+\chi_{e^*}-\chi_e +\chi_g \in \mathbb{B}_f(d).\] 

We will first show that $\vec x -\chi_{e}+\chi_{g} \notin \B_f(d)$.
Assume by  contradiction that
$\vec x -\chi_{e}+\chi_{g} \in \B_f(d)$.
We then obtain
\begin{align}
C_e^-(x_e;t_e) = C_e^-(y_e;t_e') >C_g^+(y_g;t'_g) = C_g^+(y_g;t_g) \geq C_g^+(x_g;t_g),	
\end{align}
where for the first identity we used $e \in E \setminus\{e^*,g^*\}$, the first inequality used \eqref{eq:con}, the second identity used $g \neq e^*$ and the last inequality used $g \neq e^*$, discrete convexity and $y_g \geq x_g$ for all $g \in E \setminus \{e^*\}$. This implies that 
$\vec x$ was not optimal for $P(\vec t, d)$, a contradiction.

We conclude that $\vec x - \chi_e + \chi_g \notin \B_f(d)$, but $\vec y' = \vec x - \chi_e + \chi_g - \chi_{g^*} + \chi_{e^*} \in \B_f(d)$. Thus, there exists some set $S\subset E$ with
$x(S)=f(S)$ and 
\begin{equation}\label{eq6}
\{g,g^*\}\subseteq S \quad \mbox{ and} \quad \{e,e^*\}\cap S=\emptyset.
\end{equation}
It follows that $x(S)= y'(S)=f(S)$.

We proceed to show \ref{it:condition_a_in32}. For the sake of a contradiction, suppose that
$\vec x-\chi_e + \chi_{g^*}\not\in \mathbb{B}_f(d)$. Using $g \in D_e(\vec y)$ and, thus, $x_e = y_e \geq 1$, this implies the existence of a set $T\subset E$ with $x(T)=f(T)$, $g^*\in T$, and $e\not\in T$. Since $\vec y=\vec x-\chi_{e^*} + \chi_{g^*}\in \mathbb{B}_f(d)$, it follows that $e^*\in T$.
Moreover, since $\vec y'= \vec y-\chi_e +\chi_g \in \mathbb{B}_f(d)$ and $e\not\in T$, we have $g\not\in T$. Hence, $ y'(T)=x(T)=f(T)$.
By submodularity of~$f$, the identities $x(S)= y'(S)=f(S)$ and $x(T)=y'(T)=f(T)$ imply $x(S\cap T)=y'(S\cap T)=f(S\cap T)$ and $x(S\cup T)= y'(S\cup T)=f(S\cup T)$.
Together with $g^*\in S\cap T$ and $\{e,e^*,g\}\cap (S\cap T)= \emptyset$, we arrive at the desired contradiction
\begin{equation}
f(S\cap T)=x(S\cap T) = y'(S\cap T)-1= f(S\cap T)-1.
\end{equation}
Thus, $x-\chi_e+ \chi_{g^*}\in \mathbb{B}_f(d)$, as claimed.

In a similar way, we show \ref{it:condition_b_in32}.
For the sake of a contradiction, suppose  $\vec x-\chi_{e^*}+ \chi_{g}\not\in \mathbb{B}_f(d)$.
Since $D_{e^*}(\vec x)\neq\emptyset$ and, thus, $x_{e^*} \geq 1$, this implies the existence of a set
$U \subset E$ with $x(U)=f(U)$, $g\in U$, and $e^*\not\in U$.
Since $\vec y'=x-\chi_{e^*}+\chi_{g^*}-\chi_e+\chi_g \in \mathbb{B}_f(d)$, we further have $e \in U$, and $g^*\not\in U$, implying $y'(U)=f(U)$ and $g\in S\cap U$.
Just as for \ref{it:condition_a_in32}, $x(S)=y'(S)=f(S)$ and $x(U)=y'(U)=f(U)$ leads to the desired contradiction
$$f(U\cap S)=x(U\cap S)= y'(U\cap S)-1 =f(U\cap S)-1.$$
We have treated all cases and the theorem follows.
%
\end{proof}

In a similar manner, it can be shown that at most a single local improvement step suffices to
obtain a new optimal solution for any parameter shift of type $\vec t'=\vec t-\chi_{e^*}$.
We conclude with the following corollary.

\begin{corollary}
\label{cor:shift}
Let $P(\vec t, d)$ be regular convex optimization problem.
Then, for every optimal solution $\vec x^*(\vec t,d)$ to $P(\vec t, d)$
and every $\vec t'$ with $\norm{\vec t -\vec t'}=1$, there is an optimal solution
$\vec x^*(\vec t',d)$ for $P(\vec t', d)$ satisfying:
\begin{align}\label{sens:lem}\norm{\vec x^*(\vec t,d)-\vec x^*(\vec t',d)}\leq 2.
\end{align}
\end{corollary}

We now turn to the impact
of a change of the parameter $d$ on optimal solutions of $P(\vec t, d)$.
\begin{theorem}\label{thm:dem}
Let $D\subset \N$  denote the set of possible integral values of $d$.
Let $P(\vec t, d)$ be a family of regular optimization problems with $\B_f(d)\neq \emptyset$
for all $d\in D$.

Then, for every $ d,d'\in D$ with $|d - d'|=1$ and every
optimal solution $\vec x^*(\vec t,d)$ to $P(\vec t, d)$, there is an optimal solution
$\vec x^*(\vec t,d')$ for $P(\vec t, d')$ with
\begin{align}\label{sens:demand}\norm{\vec x^*(\vec t,d)-\vec x^*(\vec t,d')}\leq  |d-d'|=1.
\end{align}
\end{theorem}
\begin{proof}
We here only prove the case $d'=d+1$, the case $d'=d-1$ follows similarly.
Define the set of resources with \emph{slack} with respect to $\vec x:= \vec x^*(\vec t,d)$
as
\[ \mathcal{S}(\vec x):=\{e \in E \mid \vec x +\chi_e \in \B_f(d+1)\}.\] 
Since by assumption $d+1\in D$ and thus $\B_f(d+1)\neq \emptyset$, we obtain $\mathcal{S}(\vec x)\neq \emptyset$
(cf. Fujishige~\cite[Theorem~2.3, pp. 35]{fujishige2005submodular}). 

We claim that the solution $\vec y := \vec x +\chi_{g^*}$ is optimal for problem $P(\vec t, d')$,
where 
\[ g^*\in \arg\min\nolimits_{g \in \mathcal{S}(\vec x)} C^+_g(x_g;t_g).\]
To prove the claim, we show that the optimality conditions of Theorem~\ref{thm:fuj} are satisfied, i.e., $C_e^-(y_e;t_e) \leq \Delta_e(\vec y; \vec t)$ for all $e \in E$. This is trivial if $D_e(\vec y) = \emptyset$. Otherwise, we distinguish two cases. If $g^* \in \arg\min_{g \in D_e(\vec y)} \{C^+_g(y_g;t_g)\}$, we obtain,
\begin{align*}
\Delta_e(\vec y; \vec t) &= C^+_{g^*}(y_{g^*};t_{g^*}) \geq C_{g^*}^+(x_{g^*};t_{g^*}) \geq C_e^-(x_e;t_e) = C_e^-(y_e;t_e)
\intertext{where for the first inequality, we used discrete convexity and for the second inequality, we used the optimality of $\vec x$ for $\B_f(d)$. If $g^* \notin \arg\min_{g \in D_e(\vec y)} \{C^+_g(y_g;t_g)\}$, we obtain}
\Delta_e(\vec y;\vec t) &= \min_{g \in D_e(\vec y)} \{C_g^+(y_g;t_g)\}	 \geq \min_{g \in D_e(\vec x)} \{C_g^+(y_g;t_g)\} = \min_{g \in D_e(\vec x)}\{C^+_g(x_g;t_g)\}\\
 &\quad\geq C^-_e(x_e;t_e) = C^-_e(y_e;t_e),	
\end{align*}
where for the first inequality we used $D_e(\vec x)\supseteq D_e(\vec y)$ for all
$e \in E\setminus \{g^*\}$ and for the second inequality we used the optimality of $\vec x$ for $\B_f(d)$.
\end{proof}

By inductively applying Theorem~\ref{thm:shift} and Theorem~\ref{thm:dem} we arrive at the following result. 
\begin{theorem}\label{thm:gen}
Let $D\subset \N$  denote the set of possible integral values of $d$ and
let $P(\vec t, d), d\in D$ be a set of regular optimization problem with 
$\B_f(d)\neq \emptyset$ for all $d\in D$.
Then, for every optimal solution $\vec x^*(\vec t,d)$ of $P(\vec t, d)$, $d\in D$, every $ d'\in D$  and every $\vec t'$, there is an optimal solution
$\vec x^*(\vec t',d')$ of $P(\vec t', d')$ with
\begin{align}\label{sens:full}\norm{\vec x^*(\vec t,d)-\vec x^*(\vec t',d')}\leq  2\norm{\vec t - \vec t'}+|d-d'|.
\end{align}
\end{theorem}
Note that for the proof, we can safely use induction on $d$ also for those $d$'s with $d_1,d_2\in D, d_1\leq d\leq d_2, d\notin D$, because $\B_f(d_i)\neq \emptyset, i=1,2$ implies $\B_f(d)\neq \emptyset$.
Moreover, regularity of $P(\vec t, d)$ does not depend on $d$.

The proofs of Theorem~\ref{thm:shift} and Theorem~\ref{thm:dem} also show 
that after a parameter change, a new optimal solution can be recovered by elementary exchange
steps, that is, by iteratively shifting one unit from one element to another, or, by adding one unit to an element.
We can summarize this discussion as follows.
\begin{corollary}
Let $D\subset \N$  denote the set of possible integral values of $d$.
Let $P(\vec t, d), d\in D$ be a set of regular optimization problems with 
$\B_f(d)\neq \emptyset$ for all $d\in D$.  Then, for every optimal solution $\vec x^*(\vec t,d)$ of $P(\vec t, d)$, $d\in D$, every $ d'\in D$  and every $\vec t'$, there is an optimal solution
$\vec x^*(\vec t',d')$ for $P(\vec t', d')$ that can be computed from
$\vec x^*(\vec t,d)$ by performing at most 
$\norm{\vec t - \vec t'}+|d-d'|$ elementary exchange steps.
\end{corollary} 

In the following section, we apply the above sensitivity results to
a quite general class of non-cooperative games, thereby establishing an
existence and computability result of pure Nash equilibria.

\section{Noncooperative Games on Polymatroids}
\label{sec:game}

We consider the following class of games. There
is a finite set $N=\{1,\dots,n\}$ of players and a finite set $E = \{1,\dots,m\}$ of elements. As it is standard in the congestion game literature, in this section we refer to the elements $e\in E$ as \emph{resources}. Each player~$i$ is associated with a demand $d_i \in \N$ and an integral polymatroid rank function $f_i : 2^E \to \N$ that together define a $d_i$-truncated integral polymatroid $\P_{f_i}(d_i)$ with base polytope $\B_{f_i}(d_i)$. A strategy of
player~$i \in N$ is to choose a vector $\vec x_i  = (x_{i,e})_{e \in E} \in \B_{f_i}(d_i)$, i.e., player~$i$ chooses an integral resource consumption $x_{i,e} \in \N$ for each resource $e$
such that the demand $d_i$ is exactly distributed among the resources and for each $U \subseteq E$ not more than $f_i(U)$ units of demand are distributed to the resources contained in $U$.
Using the notation $\vec x_i=(x_{i,e})_{e \in E}$, the set $X_i$ of feasible strategies of player~$i$ is
defined as
\begin{align*}
X_i = \B_{f_i}(d_i)
= \Bigl\{\vec x_i \in \N^{E} \mid x_{i}(U) \le f_i(U) \text{ for all } U \subseteq 
E, x_{i}(E) =d_i \Bigr\},
\end{align*}
where, for a set $U \subseteq E$, we write $x_i(U)$ shorthand for $\sum_{e \in U} x_{i,e}$.
The Cartesian product $X = X_1 \times X_2 \times \dots \times X_n$ of the players' sets of feasible
strategies is the joint strategy space. An
element $\vec x = (\vec x_i)_{i \in N} \in X$ is a strategy profile. For a resource $e$, and a strategy profile $\vec x 
\in X$, we write $x_e = \sum_{i \in N}
x_{i,e}$ and $x_{-i,e} = \sum_{j \in N\setminus\{i\}}
x_{j,e}$. The private cost of
player~$i$ under strategy profile $\vec x \in X$ is defined as 
\[\pi_i(\vec x) = \sum_{e \in E}
C_{i,e}(x_{i,e};x_{-i,e}).\]
We assume that every $C_{i,e}$ fulfills the conditions of Assumption~\ref{ass:cost-general}.
In the remainder of the paper, we will compactly
represent the strategic game by the tuple 
$ G=(N,X,(d_i)_{i\in N},(C_{i,e})_{i\in N,e \in E}).$
We use standard game theory notation. For a player $i \in N$ and a strategy profile $\vec x \in X$,
we write $\vec x$ as $(\vec x_i,\vec x_{-i})$. A \emph{best response} of player~$i$ to $\vec x_{-i}$ is a strategy $\vec 
x_i \in X_i$ with $\pi_i(\vec x_i, \vec{x}_{-i}) \leq \pi_i(\vec y_i, \vec{x}_{-i})$ for all $\vec y_i \in X_i$.
A pure Nash equilibrium is a strategy profile $\vec x \in X$ such that for each player~$i$ the strategy $\vec x_i$ is a 
best response to $\vec x_{-i}$.

\subsection{Notable Special Cases}
\label{apps:examples}

We proceed to illustrate that we obtain the well known classes of integer-splittable singleton congestion games and matroid congestion games as special cases of congestion games on integer polymatroids.

\subsubsection{Singleton Integer-Splittable Congestion Games}
Tran-Thanh et al.~\cite{TranPCRJ11} consider singleton integer-splittable congestion games where each player~$i$ is associated with an integral demand $d_i \in \N$ that needs to be distributed integrally over a player-specific subset $E_i \subseteq E$ of resources. Every resource has a non-decreasing and convex cost function $c_e : \N \to \R_+$ and the private cost of a player~$i$ is equal to
\begin{align*}
\pi_i(\vec x) = \sum_{e \in E} c_e(x_{i,e}+x_{-i,e})x_{i,e}.
\end{align*}

We proceed to show that this class of games is contained in the class of polymatroid games as a special case. 

\begin{proposition}
Singleton integer-splittable congestion games are polymatroid games.
\end{proposition}

\begin{proof}
For $i \in N$ and $U \subseteq E$, we let
\[
f_i(U) =
\begin{cases}
d_i, &\text{if $U \cap E_i \neq \emptyset$}\\
0,	&\text{otherwise}. 	
\end{cases}
\]
For $i \in N$ and $e \in E$, we let $C_{i,e}(x_{i,e};x_{-i,e}) = c_e(x_{i,e}+x_{-i,e})x_{i,e}$. 
First, we show that the functions $f_i$ are normalized, monotone and submodular. For submodularity, it suffices to show that $f(U \cup \{v\})-f(U) \geq f(V \cup \{v\}) -f(V)$ for all $U \subseteq V$ and $v \notin U$. The inequality can only be violated if $f(V \cup \{v\}) -f(V) = d_i$ which implies $v \in E_i$ and $V \cap E_i = \emptyset$. This, however, implies $U \cap E_i = \emptyset$ and, thus, $f(U \cup \{v\})-f(U) = d_i$.

We proceed to show that the cost functions are regular provided that $c_e$ is non-decreasing and convex. Discrete convexity is easy to verify. For regularity, we compute for arbitrary $i \in N$ and $e \in E$
\begin{align*}
C_{i,e}^-(x;t+1)& = c_e(x+t+1)x - c_e(x+t)(x-1)\\
& \leq c_e(x+t+1)(x+1) - c_e(x+t)x \\
&= C_{i,e}^-(x+1;t),
\end{align*}
where the inequality follows since $c$ is non-decreasing.
\end{proof}

\subsubsection{Matroid Congestion Games with Player-Specific Costs}
Ackermann et al.~\cite{Ackermann09} studied matroid congestion games with player-specific costs, where each player $i$ is associated with a \emph{matroid} $M_i=(E_i,\I_i,)$
defined on some player-specific subset $E_i\subseteq E$. The strategy space for
every $i\in N$ is equal to the set $\mathcal{B}_i$ of bases of $M_i$. Given a strategy profile $(B_1,\dots,B_n)$ with $B_j \in \mathcal{B}_j$ for all players~$j$, the private cost of player~$i$ is defined as
\begin{align*}
\pi_i(B_1,\dots,B_n) = \sum_{e \in B_i} c_{i,e}(|\{j \in N : e \in B_j\}|),
\end{align*}
where the functions $c_{i,e} : \N \to \R_+$ with $i\in N$ and $e \in E$ are non-decreasing. We proceed to show that this class of games is contained in the class of polymatroid games as a special case.

\begin{proposition}
Matroid congestion games are polymatroid games.
\end{proposition}

\begin{proof}
For a player~$i$ we associate with each basis $B_i \in \mathcal{B}_i$ its characteristic vector $\vec x_i(B_i) = (x_{i,e}(B_i))_{e \in E}$ defined as
\[
x_{i,e}(B_i) =
\begin{cases}
1, &\text{if $e \in B_i$}\\
0, &\text{otherwise.}	
\end{cases}
\]
It is well known, that for each matroid $M_i = (E_i,\I_i)$, there is a function $\text{rk} : E_i \to \N$, called the rank function of the matroid, such that
\[
\{\vec x_i(B_i) : B_i \in \mathcal{B}_i\} = \{\vec x_i \in \N^{E_i} : x_i(E_i) = \text{rk}(E_i) \text{ and } x_i(U) \leq \text{rk}(U) \text{ for all } U \subseteq E_i\}.
\]
Moreover, the rank function is normalized, monotone and submodular. We let $f_i(U) = \rk(U \cap E_i)$ and let $d_i = \rk(E_i)$. Then, $f_i$ is submodular since for all $U,V \in 2^E$ we have
\begin{align*}
f_i(U) + f_i(V) &= \text{rk}(U \cap E_i) + \text{rk}(V \cap E_i)\\
&\geq \text{rk}((U \cap E_i) \cup (V \cap E_i)) + \text{rk}(U \cap V \cap E_i)\\
&= f_i(U \cup V) + f_i(U \cap V)
\end{align*}
where the inequality uses the submodularity of the rank function.

For $i \in N$ and $e \in E$, let 
\[
C_{i,e}(x;t) =  
\begin{cases}
c_{i,e}(x + t), &\text{ if $x_{i,e}$ = 1}\\
0, &\text{ if $x_{i,e}$ = 0},	
\end{cases}
\]
and note that the rank function is subcardinal, i.e., $\text{rk}(U) \leq |U|$ for all $U \subseteq E$, so that $\text{rk}(\{r\}) \leq 1$ for all $e \in E$ and, thus, $C_{i,e}$ is well-defined. As a consequence, we need to require regularity and discrete convexity only for $x \in \{0,1\}$. As for regularity, we obtain
\begin{align*}
C_{i,e}(0;t+1) = 0 \leq c_{i,e}(1+t) = C_{i,e}(1;t)
\end{align*}
for all $t \in \N$ by the non-negativity of $c_{i,e}$. As for discrete convexity, we do not need to require discrete convexity of the function $x \mapsto C(x;t)$ as $x$ only takes two different values. Moreover, since $C_{i,e}(0;t) =0$ for all $t$ we only have to check that $t \mapsto C_{i,e}(1;t)$ is discrete convex. To this end, we calculate
\begin{align*}
C_{i,e}^-(1;t) &= C_{i,e}(1+t_e) - c_{i,e}(1+t_e-1)\\
&\leq c_{i,e}(1+t_e+1) - c_{i,e}(1+t_e),
\end{align*}
where we used that $c_{i,r}$ is non-decreasing.
\end{proof}

\section{Equilibrium Existence}
In this section, we give an algorithm that computes
a pure Nash equilibrium for polymatroid games.
Our algorithm relies on the two sensitivity results stated in  Theorem~\ref{thm:shift}
and Theorem~\ref{thm:dem}.

\subsection{The Algorithm}
Both sensitivity results are used as the main
building blocks for Algorithm~\ref{alg:greedy}
that computes a pure Nash equilibrium
for congestion games on integral polymatroids.
Algorithm~\ref{alg:greedy}
maintains \emph{preliminary demands, strategy spaces, and strategies} of the players
denoted by $\bar d_i\leq d_i$, $\bar X_i=X_i(\bar d_i)$, and $\vec x_i\in \bar X_i$, respectively. Initially,  $\bar{d}_i$ is set to zero for all $i\in N$ and the strategy profile, where the strategy of each player equals the zero vector is a pure Nash equilibrium for this game in which the demand of each player is zero.

Then, in each round, for some player~$i$ the demand is increased from 
$\bar d_i$ to $\bar d_i+1$, and a best response $\vec y_i\in X(\bar d_i+1)$ with $\norm{\vec x_i- \vec y_i}=1$ is computed, see Line~\ref{it:comp_demand} in Algorithm~\ref{alg:greedy}. By Theorem~\ref{thm:dem},
such a best response always exists. In effect, the load on exactly one resource $e$
increases and only those players~$j$ with $x_{j,e} > 0$ on this
resource can potentially decrease their private cost by a deviation.
By Theorem~\ref{thm:shift}, a best response of such players  consists w.l.o.g. of moving a single unit from this resource to another resource, see Line~\ref{it:choose_yi} of Algorithm~\ref{alg:greedy}.
As a consequence, during the while-loop (Lines~\ref{it:if}-\ref{it:endif}),
only one additional unit (compared to the previous iteration)
is moved preserving the invariant that only players using a resource
to which this additional unit is assigned may have an incentive to profitably deviate.
Thus, if the while-loop is left, the current strategy profile~$\vec x$
is a pure Nash equilibrium for the reduced game $\bar{G} = (N, \bar X, \bar d,(C_{i,e})_{i \in N, e \in E})$.

\begin{algorithm}[tb]
 \caption{Compute PNE}
 \label{alg:greedy}
 \Indm\Indmm
    \KwIn{$G=(N,X,(d_i)_{i\in N},(C_{i,e})_{i\in N,e \in E})$}
  \KwOut{pure Nash equilibrium $\vec x$}
  \Indp\Indpp
$\bar d_i \leftarrow 0, \bar X_i\leftarrow X_i(0)$ and $\vec x_i\leftarrow \vec 0$ for all $i\in N$\;
		\For{$k=1,\dots, \sum_{i\in N}d_i$}{	
			Choose $i\in N$ with $\bar d_i<d_i$\;\label{it:choose_player}
			$\bar d_i\leftarrow \bar d_i+1$; $\bar X_i\leftarrow X_i(\bar d_i)$\;\label{it:set_increase_demand}
	Choose a best response $\vec y_i \in \bar{X}_i$ with $\norm{\vec y_i-\vec x_i} = 1$\;\label{it:comp_demand}
$\vec x_i \leftarrow  \vec y_i$\; \label{it:increase_demand}
\While{$\exists i\in N$ who can improve in $\bar G = (N, \bar X, \bar d,(C_{i,e})_{i \in N, e \in E})$
\label{it:if}}
 {Compute a best response $\vec y_i \in \bar X_i$ with $\norm{\vec y_i- \vec x_i}=2$\;  \label{it:choose_yi}
$\vec x_i \leftarrow  \vec y_i$\;}
\label{it:endif}}
Return $\vec x$\;
\end{algorithm}

Now we are ready to prove the main existence result.
\begin{theorem}
\label{thm:main}
Polymatroid games possess a pure Nash equilibrium.
\end{theorem}
\begin{proof}
We prove by induction on the total demand $d = \sum_{i\in N}d_i$ of the input game $G = (N, X, (d_i)_{i \in N}, 
(C_{i,e})_{i \in N, e \in E})$ that Algorithm~\ref{alg:greedy}
computes a pure Nash equilibrium of $G$.

For $d = 0$, this is trivial.
Suppose that the algorithm works correctly for games with total demand $d-1$ for some $d \geq 1$ and consider a game $G$ with total demand $d$.
Let us assume that in Line~\ref{it:choose_player}, the algorithm always chooses a player with minimum index. Consider the game $G' = (N,X,(d'_i)_{i\in N}, (C_{i,e})_{i \in N, e \in E})$ that differs from $G$ only in the fact that the demand of the last player~$n$ is reduced by one, i.e. $d'_i = d_i$ for all $i<n$ and $d_n' = d_n - 1$. Then, when running the algorithm with $G'$ as input, the $d-1$ iterations (of the for-loop) are equal to the first $d-1$ iterations when running the algorithm with $G$ as input. Thus, with $G$ as input, we may assume that after the first $d-1$ iterations, the preliminary strategy profile that we denote by $\vec x'$ is a pure Nash equilibrium of $G'$.

We analyze the final iteration $k=d$ of the algorithm in which the demand of player~$n$ is increased by $1$ (see Line~\ref{it:set_increase_demand}).
In Line~\ref{it:comp_demand}, a best reply $\vec y_n$ with $\norm{\vec x_n-\vec y_n} = 1$
is computed which exists by Lemma~\ref{thm:dem}. Then, as long as there is a player~$i$ that can improve unilaterally, in Line~\ref{it:choose_yi}, a best response $\vec y_i$ with $\norm{\vec y_i- \vec x_i} = 2$ is computed which exists by Lemma~\ref{thm:shift}.

It remains to show that the while-loop in Lines~\ref{it:if}--\ref{it:endif} terminates.
To prove this, we give each unit of demand of each player~$i \in N$ an identity
denoted by $i_j, j=1,\dots,d_i$. For a strategy profile $\vec x$, we define
$r(i_j,\vec x)\in E$ to be the resource to which  unit $i_j$ is assigned in strategy profile $\vec x$.
Let $\vec x^l$ be the strategy profile after Line~\ref{it:choose_yi} of the algorithm has been executed the $l$-th time, where we use the convention that $\vec x^0$ denotes the preliminary strategy profile when entering the while-loop. As we chose in Line~\ref{it:comp_demand} a best reply $\vec y_n$ of player~$n$ with $\norm{\vec x_n-\vec y_n} = 1$, there is a unique resource $e_0$ such that $x_{e_0}^0 = x'_{e_0} + 1$ and $x_{e}^0 = x_{e}'$ for all $e \in E \setminus \{e_0\}$. Furthermore, because we choose in Line~\ref{it:choose_yi}
a best response $\vec y_i$ with $\norm{\vec y_i- \vec x_i} = 2$, a simple inductive claim shows that after each iteration~$l$ of the while-loop, there is a unique resource $e_l \in E$ such that $x_{e_l}^l = x_{e_l}' + 1$ and $x_{e}^l = x_{e}'$ for all $e \in E \setminus \{e_l\}$.

For any $\vec x^{l}$ during the course of the algorithm, we define the
\emph{marginal cost} of unit $i_j$
under strategy profile $\vec x^{l}$ as
\begin{align}\label{def:delta}
\Delta_{i_j}(\vec x^l)=
\begin{cases} 
C_{i,e}^-(x_{i,e}^l;x_{-i,e}^l), & \text{if }e=e(i_j,\vec x^{l})=e_{l}\\
C_{i,e}^-(x_{i,e}^l;x_{-i,e}^l+1),
& \text{if }e=e(i_j,\vec x^{l})\neq e_{l}.\end{cases}
\end{align}
Intuitively, if $e(i_j,\vec x^{l})=e_{l}$, the value $\Delta_{i_j}(\vec x^{l})$ measures the \emph{cost saving} on resource 
$e(i_j,\vec x^{l})$ if $i_j$ (or any other unit of player $i$ on resource $e(i_j,\vec x^{l})$)
is removed from $e(i_j,\vec x^{l})$.
If $e(i_j,\vec x^{l})\neq e_{l}$,
the value $\Delta_{i_j}(\vec x^{l})$ measures
the cost saving if  $i_j$ is removed from $e(i_j,\vec x^{l})$ after the total allocation 
has been increased by one unit by some other player.
For a strategy profile $\vec x$
we define $\Delta(\vec x)=(\Delta_{i_j}(\vec x))_{i=1,\dots,n, j=1,\dots, d_i}$ to be the vector
of marginal costs and let
$\bar \Delta(\vec x)$ be the vector of marginal costs sorted in non-increasing order.
We claim that $\bar \Delta(\vec x)$ decreases lexicographically
during the while-loop.
To see this, consider an iteration $l$ in which some unit $i_j$
of player $i$
is moved from resource $e_{l-1}$ to resource $e_{l}$.

For proving $\bar \Delta(\vec x^{l})<_{\text{lex}} \bar \Delta(\vec x^{l-1})$,
we first observe that 
we only have to care for $\Delta$-values that correspond to units $i_j$ of the deviating player~$i$, because for all players $h\neq i$ we obtain
$\Delta_{h_j}(\vec x^{l-1})=\Delta_{h_j}(\vec x^{l})$ 
for all $j=1,\dots, d_h$.
This follows immediately if 
$h_j$ is neither assigned to $e_{l-1}$ nor to $e_{l}$.
If $h_j$ is assigned to $e_{l-1}$ or $e_{l}$,
then we switch the case in \eqref{def:delta},
and the claimed equality still holds.
It remains to consider the $\Delta$-values corresponding to the units of the deviating player~$i$. Recall that the deviation of player~$i$ consists of moving unit $i_j$ from resource $e_{l-1}$ to resource $e_{l}$.
We obtain
\begin{align*}
\Delta_{i_j}(\vec x^{l-1}) &=C_{i,e_{l-1}}^-(x_{i,e_{l-1}}^{l-1};x_{-i,e_{l-1}}^{l-1})
\\&>C_{i,e_{l}}^+(x_{i,e_{l}}^{l-1};x_{-i,e_{l}}^{l-1})\\
&=C_{i,e_{l}}^-(x_{i,e_{l}}^{l-1}+1;x_{-i,e_{l}}^{l-1})\\
&=C_{i,e_{l}}^-(x_{i,e_{l}}^{l};x_{-i,e_{l}}^{l})\\
& =  \Delta_{i_j}(\vec x^{l}),
\end{align*}
where the strict inequality follows since player $i$ strictly improves.
For every unit $i_m$ of player~$i$ that is assigned to resource $e_l$ as well, i.e, $ e(i_m,\vec x^{l})= e(i_j,\vec x^{l})=e_{l}$, we have
 $\Delta_{i_j}(\vec x^{l})=\Delta_{i_m}(\vec x^{l})$
 since the $\Delta$-value is the same for all units of a single player
assigned to the same resource. The $\Delta$-values
of such units $i_m$ might have increased, but only to the $\Delta$-value of unit $i_j$.

Next, consider the $\Delta$-values of a unit $i_m$ assigned to resource $e_{l-1}$, i.e., $ e(i_m,\vec x^{l})= e(i_j,\vec x^{l-1})=e_{l-1}$.
We obtain
\begin{align*}
\Delta_{i_m}(\vec x^{l})&=C_{i,e_{l-1}}^-(x_{i,e_{l-1}}^{l};x_{-i,e_{l-1}}^{l}+1)\\
&=C_{i,e_{l-1}}^-(x_{i,e_{l-1}}^{l-1}-1;x_{-i,e_{l-1}}^{l-1}+1)\\
&\leq C_{i,e_{l-1}}^-(x_{i,e_{l-1}}^{l-1};x_{-i,e_{l-1}}^{l-1})
\\ & = \Delta_{i_m}(\vec x^{l-1}),
\end{align*}
where for the inequality we used that $C_{i,e_{l-1}}$ is regular.
Altogether, the $\Delta$-values of all units of all players $h\neq i$
have not changed, for player $i$, the  $\Delta$-values of remaining units assigned to resource $e_{l-1}$
 decreased, and the $\Delta$-values assigned to resource $e_{l}$ increased exactly to $\Delta_{i_j}(\vec x^{l})$
which is strictly smaller than $\Delta_{i_j}(\vec x^{l-1})$.
Thus, $\bar \Delta(\vec x^{l})<_{\text{lex}} \bar \Delta(\vec x^{l-1})$ follows.
\end{proof}

The following corollary states an upper bound on the number of iterations of the algorithm in terms of $\delta = \max_{i \in N} d_i$.

\begin{corollary}\label{cor:runtime}
The number of iterations is at most $n^2 m \dmax^{3}$. 
\end{corollary}

We analyze the worst-case runtime of Algorithm~\ref{alg:greedy}. To this end, let us consider the iterations of the algorithm for fixed $k$. In the proof of Theorem~\ref{thm:main}, we showed that for fixed $k$, for each player, the sorted vector of marginal costs (as defined in \eqref{def:delta}) decreases lexicographically during the while-loop. Moreover, the marginal cost of a particular unit of demand $i_j$ of player~$i$ assigned to a resource $e$ does not depend on the aggregated demand $\sum_{j \in N} x_{j,e}$ of all players for resource $e$, but only on the number of units of demand $x_{i,e}$ assigned to $e$ by player~$i$. We derive that for each player~$i$ and each resource $e$ at most $d_i$ different marginal cost values can occur. This implies that each unit of demand of player~$i$ visits each resource at most $d_i$ times. Thus, the total number of iterations of the while-loop is bounded by $\sum_{i \in N} (m \cdot d_i^2)$. Setting $\dmax = \max_{i \in N} d_i$, this expression is bounded by $n m \dmax^2$, where $n = |N|$. Using that there are $\sum_{i \in N} d_i \leq n \cdot \dmax$ iterations of the for-loop, the claimed result follows.

\section{Non-Polymatroid Regions}
\label{sec:violation}
The proofs of the results obtained in \S~\ref{sec:sensitivity} and \S~\ref{sec:game} relied on the fact that the function $f$ is submodular and, thus, the feasible region of the optimization problem and the strategy spaces of the players, respectively, are polymatroids. One may wonder whether polymatroids are the maximal combinatorial structure for which these results hold. In this section, we will give an affirmative answer to this question. In fact, we will work towards showing that for \emph{every} normalized, monotonic and non-submodular function $f$, there is a convex and regular optimization problem with feasible set
\begin{align}
\label{eq:strategy_space}
\B_{f}(d)= \Bigl\{\vec  x\in \N^E \ \mid x(U)\le f(U) \text{ for all } U\subseteq E,\, x(E) = d \Bigr\}
\end{align}
with $d \in \N$ such that the sensitivity results of Theorem~\ref{thm:shift} and Theorem~\ref{thm:dem} do not hold. Moreover, there is a game with convex and regular cost functions where the players' strategies are isomorphic to \eqref{eq:strategy_space} that has not a pure Nash equilibrium. This implies that also for the existence result of Theorem~\ref{thm:main} the polymatroid structure is maximal.

For ease of exposition, we assume that $f$ is strictly positive in the sense that $f(U) > 0$ for all $U \in 2^E \setminus \emptyset$. This assumption can be made without loss of generality since a non-empty set of resources $U$ with $f(U) = 0$ is not used in any strategy anyway and, thus, can effectively be removed from $E$. 

Formally, let
\begin{align*}
\mathcal{X}(d) &= \Bigl \{\B_f(d) \mid f \text{ is strictly positive, normalized and monotonic} \Bigr\},\\
\mathcal{X}^*(d) &= \Bigl \{\B_f(d) \mid f \text{ is strictly positive, normalized, monotonic, and submodular} \Bigr\}
\end{align*}
denote the feasible regions that can be described by arbitrary and submodular functions $f$, respectively. 

First, note that $\mathcal{X}(1) = \mathcal{X}^*(1)$. To see this, note that $f(\{e\}) \geq 1$ as $f$ is strictly positive, so $f$ does not encode any constraints on where the single unit of demand can be put. Thus, the set of strategies can be described equivalently by a function $f'$ with $f'(U) = 1$ for all $U \in 2^E \setminus \{\emptyset\}$. It is straightforward to verify that $f'$ is submodular.

More interestingly, already for $d=2$, the feasible regions contained in $\mathcal{X}(2)$ and $\mathcal{X}^*(2)$ differ. 
We start with the following observations regarding the feasible regions in $\mathcal{X}(2) \setminus \mathcal{X}^*(2)$.
\begin{lemma}
\label{lem:non-triviual}
For any $X \in \mathcal{X}(2) \setminus \mathcal{X}^*(2)$ there are  $f : 2^E \to \N$  and $S,T \in 2^E$ such that
\begin{enumerate}
	\item $X = \B_f(2)= \bigl\{\vec  x\in \N^E \ \mid x(U)\le f(U) \text{ for all } U\subseteq E,\, x(E) = 2 \bigr\}$,
	\item For any constraint $x(U) \leq f(U)$, there is $\vec x \in \B_f(2)$ with $x(U) = f(U)$.\label{it:lemma_second_point}
	\item $f(S) = f(T) = f(S \cap T) =1$ and $f(S \cup T) = 2$.\label{it:lemma_third_point}
\end{enumerate}
\end{lemma}

\begin{proof}
Since $X \in \mathcal{X}(2)$, there is a strictly positive, normalized and monotonic function $f$ with $X = \B_f(d)$. Note that $f$ is not submodular since $X \in \mathcal{X}^*$(2), otherwise.

To prove {\it\ref{it:lemma_second_point}.}, suppose that there is $U' \subseteq R$ such that $x(U') < f(U')$ for all $\vec x \in \B_f(2)$. Consider the new function $f' : 2^E \to \N$ defined as
\begin{align*}
f'(U) = \begin{cases}
 f(U), & \text{ if } U \neq U',\\
 f(U) - 1, & \text{ if } U = U'.	
 \end{cases}
\end{align*}
By construction $\B_f(2) = \B_{f'}(2)$. Applying the above argumentation on $f'$, we derive that $f'$ is not submodular as well. Decreasing the right-hand side of non-tight constraints iteratively in this manner, we finally obtain a non-submodular function $f : 2^E \to \N$ with $X = \B_f(2)$ such that for any constraint of type $x(U) \leq f(U)$, there is $\vec x \in \B_f(2)$ with $x(U) = f(U)$.

To prove {\it\ref{it:lemma_third_point}}, first note that $x(E) = 2$ together with the second statement of the lemma implies that $f(U) \leq 2$ for all $U \subseteq E$. Since $f$ is not submodular, there are sets $S, T \in 2^E$ such that
\begin{align}
\label{eq:non-submodular}
	f(S) + f(T) < f(S \cap T) + f(S \cup T).
\end{align}
It is straightforward to verify that this inequality can only be satisfied when $S$ and $T$ are non-empty. We claim that $S \cap T$ is non-empty as well. For the sake of a contradiction, let us assume that \eqref{eq:non-submodular} is satisfied by $S,T \in 2^E \setminus \emptyset$ with $S \cap T= \emptyset$. Using the second statement of the lemma, there is a vector $\vec x \in \B_f(2)$ such that the constraint $x(S \cup T) \leq f(S \cup T)$ is tight, i.e., $x(S \cup T) = f(S \cup T)$. Since $S$ and $T$ are disjoint, we obtain $x(S \cup T) = x(S) + x(T) \leq f(S) + f(T) < f(S \cup T)$, a contradiction. We have established that $S$, $T$ and $S \cap T$ are non-empty. Using the strict positivity of $f$, this implies that $f(S) \geq 1$, $f(T) \geq 1$ and $f(S \cap T) \geq 1$. Together with $f(S \cup T) \leq 2$ and the monotonicity of $f$, this implies that
\begin{align}
\label{eq:st}
f(S) &= f(T) = f(S\cap T) = 1 & & \text{ and } & f(S\cup T) &= 2,
\end{align}
as claimed.
\end{proof}

We proceed to give two additional lemmas that give further structural results regarding the sets $S,T \in 2^E \setminus \emptyset$ for which the submodularity constraints are violated. First we show that each strategy whose support is contained in $S \cup T$ does not use a resource in $S \cap T$.

\begin{lemma}
\label{lem:st}
Let $f : 2^E \to \N$ and $S,T \in 2^E$ be as guaranteed by Lemma~\ref{lem:non-triviual}. Then for any $\vec x \in \B_f(2)$ with $\support(\vec x) \subseteq S \cup T$ we have $\support(\vec x) \cap (S \cap T) = \emptyset$.
\end{lemma}

\begin{proof}
Suppose there is $\vec x \in \B_f(d)$ and a resource $e \in E$ with $\support(\vec x)\subseteq S\cup T$
and  $r \in \support(\vec x) \cap S \cap T$. Because $\vec x$ is integral and $f(S \cap T) = 1$, this implies that $x(e) = 1$. Since $x(R) = 2 = x(S \cup T)$ and $x(S \cap T) \leq 1$ there is another element $e' \neq e$
with $e' \in S \Delta T$ and $x(e') = 1$. It is without loss of generality to assume that $e' \in S \setminus T$. This, however, implies that  $x(S) \geq x(e)+ x(e') = 2$, a contradiction to $x(S)\le 1$.
\end{proof}

Combining the two previous lemmas, we derive the existence of four distinct \emph{critical elements} which are used by two vectors with disjoint supports.

\begin{lemma}
\label{lem:non-submodular-strategies}
Let $f : 2^E \to \N$ be as guaranteed by Lemma~\ref{lem:non-triviual}. Then, there are four distinct elements $e_1,e_2,e_3,e_4 \in E$ and two vectors $\vec x, \vec y \in \B_f(2)$ with the following properties:
\begin{enumerate}
	\item $x(e_1) = x(e_2) = 1$ for some $e_1 \in E \setminus (S \cup T)$ and $e_2 \in S \cap T$.
	\item $y(e_3) = y(e_4) = 1$ for some $e_3 \in S \setminus T$ and $e_4 \in T \setminus S$.
	\item For all other strategies $\vec z \in \B_f(2) \setminus \{\vec x, \vec y\}$ with $\support(\vec z) \subseteq \{e_1,e_2,e_3,e_4\}$ one of the following three cases holds:\label{it:no-submodular-third}
	\begin{enumerate}
		\item $\support(\vec z) = \{e_1,e_3\}$.
		\item $\support(\vec z) = \{e_1,e_4\}$.
		\item $\support(\vec z) = \{e_1\}$.
	\end{enumerate} 
\end{enumerate}
\end{lemma}

\begin{proof}
By Lemma~\ref{lem:non-triviual}, there is a strategy $\vec x$ for which the constraint $x(S \cap T) \leq 1$ is tight. This implies the existence of a element $e_2 \in S \cap T$ with $x(e_2) = 1$. By Lemma~\ref{lem:st}, $\support(\vec x) \not\subseteq S \cup T$ implying the existence of a element $e_1 \in E \setminus (S \cup T)$ with $x(e_1) = 1$.

By Lemma~\ref{lem:non-triviual}, there is also another strategy $\vec y$ for which the constraint $x(S \cup T) \leq 2$ is tight. First note that $y(r) \leq 1$ for all $r \in S \cup T$ as otherwise the constraint $y(S) \leq f(S) = 1$ or $y(T) \leq f(T) = 1$ would be violated. This implies the existence of two distinct elements $e_3,e_4 \in S \cup T$ such that $y(e_3) = y(e_4) = 1$. Further note that $e_3,e_4 \notin S \cap T$ as otherwise Lemma~\ref{lem:st} is violated. Using $y(S) \leq 1$ and $y(T) \leq 1$, we derive that $e_3 \in S \setminus T$ and $e_4 \in T \setminus S$.

To see the last part of the claim, note that any strategy $\vec z$ with $e_2 \in \support(\vec z) \subseteq \{e_1,e_2,e_3,e_4\}$ must have $\support(\vec z) = \{e_1,e_2\}$ as otherwise the constraint $z(S) \leq 1$ or $z(T) \leq 1$ is violated. Thus, $\vec z = \vec x$ for any such $\vec z$. Similarly, note that any strategy $\vec z$ with a singleton support $\support(\vec z) = \{e\}$ with $e \in \{e_1,e_2,e_3,e_4\}$ must have $\support(\vec z) = \{e_1\}$, as otherwise the constraint $z(S) \leq 1$ or $z(T) \leq 1$ is violated. These two observations leave only room for the strategies as in part {\it\ref{it:no-submodular-third}.} of the statement of the lemma.
\end{proof}

\subsection{Violation of Sensitivity Results (Corollary~\ref{cor:shift}) and Theorem~\ref{thm:dem}}

We first show that for any feasible region $X \in \mathcal{X}(2) \setminus X^*(2)$ that is not described by a submodular capacity constraint, the sensitivity results of Corollary~\ref{cor:shift} does not hold.

\begin{theorem}
\label{thm:violation_shift}
For any $X \in \mathcal{X}(2) \setminus X^*(2)$, there is an optimization problem of the form
\begin{align*}
\text{minimize~} &\sum_{e \in E} C_e(x_e;t_e)\\	
\text{subject to~} & \vec x \in X,
\end{align*}
and $\vec t, \vec t' \in N^E$ such that $\norm{\vec t - \vec t'} = 1$ but $\norm{\vec x^*(\vec t,2)-\vec x^*(\vec t',2)} > 2$.
\end{theorem}

\begin{proof}
By Lemma~\ref{lem:non-submodular-strategies}, there are four critical elements $e_1,e_2,e_3,e_4$ such that $X$ can be decomposed in the following way:
\begin{align*}
X = \{\vec x, \vec y\} \cup X^\text{\text{crit}} \cup X^{\text{out}},	
\end{align*}
where $\supp(\vec x) = \{e_1,e_2\}$, $\supp(\vec y) = \{e_3,e_4\}$. The set
\begin{align*}
X^{\text{crit}}	= \Bigl\{ \vec z \in X : \supp(z) \in \{\{e_1,e_3\},\{e_1,e_4\}, \{e_1\}\} \Bigr\}
\end{align*}
contains an arbitrary subset of vectors whose support is a subset of the four critical resources, but that are not contained in $\{\vec x, \vec y\}$. By Lemma~\ref{lem:non-submodular-strategies}, the only supports that can occur for vectors in $X^{\text{crit}}$ are $\{\{e_1,e_3\}$,$\{e_1,e_4\}$ and $\{e_1\}$. Finally, the set $X^{\text{out}}$ contains all vectors whose support contains a non-critical element.

Let
\begin{align*}
C_e(x_e;t_e) &= (x_e+t_e)^2 & &\quad\text{for $e \in \{e_1,e_2\}$},\\
C_e(x_e;t_e) &= (x_e+t_e)^2+1/2, & &\quad\text{for $e \in \{e_3,e_4\}$},\\
C_e(x_e;t_e) &= 20 & &\quad\text{for all $e \in E \setminus \{e_1,e_2,e_3,e_4\}$}.	
\end{align*}
and consider the parameter vectors $\vec t = \vec 0$ and $\vec t' = \chi_{e_2}$. It is easy to see that $\vec x^*(\vec t,2) = \chi_{e_1} + \chi_{e_2}$ is the unique optimal solution for parameter vector $\vec t$. On the other hand, for $\vec t'$ the unique optimal solution is $x^*(\vec t',2) = \chi_3 + \chi_4$. We note that $\norm{\vec x^*(\vec t,2) - \vec x^*(\vec t',2)} = 4$ even though $\norm{\vec t-\vec t'} = 1$ proving the claimed result. 
\end{proof}

With the same construction, it is also not hard to verify, that also Theorem~\ref{thm:dem} does not continue to hold for any feasible region that is not a polymatroid.

\begin{theorem}
For any $X \in \mathcal{X}(2) \setminus X^*(2)$, there is an optimization problem of the form
\begin{align*}
\text{minimize~} &\sum_{e \in E} C_e(x_e;t_e)\\	
\text{subject to~} & \vec x \in X,
\end{align*}
and $d, d' \in N^E$ and $\vec t \in N^E$ such that $|d - d'| = 1$ but $\norm{\vec x^*(\vec t,d)-\vec x^*(\vec t,d')} > 1$.
\end{theorem}

\begin{proof}
With the same construction as in the proof of Theorem~\ref{thm:violation_shift}, we compute that the unique optimal solution for $\vec t=\chi_{e_2}$ and $d=1$ is $\vec x^*(\chi_{e_2},1) = \chi_{e_1}$. However, as argued in the proof of Theorem~\ref{thm:violation_shift}, $\vec x^*(\chi_{e_2},2) = \chi_{e_3} + \chi_{e_4}$. We obtain $\norm{\vec x^*(\chi_{e_1},2) -  \vec x^*(\chi_{e_2},2)} = 3$ proving the claimed result. 	
\end{proof}

\subsection{Violation of the Existence of Equilibria (Theorem~\ref{thm:main})}

We proceed to show that also the existence result for pure Nash equilibria of Theorem~\ref{thm:main} does not continue to hold. In fact for any non-polymatroid structure $X \in \mathcal{X}(2) \setminus \mathcal{X}^*(2)$ there is a two-player game where both players' strategies are isomorphic to $X$ that does not habe a pure Nash equilibrium.

\begin{theorem}
\label{thm:no_nash}
For any $X \in \mathcal{X}(2) \setminus \mathcal{X}^*(2)$, there is a two-player game in which the  strategy set of both players is isomorphic to $\B_f(2)$ and that does not have a pure Nash equilibrium.
\end{theorem}

\begin{proof}
Let $f$ be as guaranteed by Lemma~\ref{lem:non-triviual}. By Lemma~\ref{lem:non-submodular-strategies}, for every player~$i \in \{1,2\}$, there are four critical resources $e^1_i,e^2_i,e^3_i,e^4_i$ such that the strategy set $X_i$ of player~$i$ can be decomposed as
\begin{align*}
X_i &= \{\vec x_i, \vec y_i\} \cup X_i^{\text{opt}} \cup X_i^{\text{out}},	
\intertext{where $\support(\vec x_i) = \{e^1_i,e^2_i\}$ and  $\support(\vec y_i) = \{e^3_i,e^4_i\}$. The set}
X_i^{\text{crit}} &= \Bigl\{\vec z_i \in X_i : \support(\vec z_i) \in \bigl\{ \{e^1_i,e^3_i \}, \{e^1_i,e^4_i\}, \{e^1_i\} \bigr\}\Bigr\}
\intertext{contains a possibly empty subset of strategies whose support is contained in the set of critical resources $\{e^1_i,e^2_i,e^3_i,e^4_i\}$, and the set}
X_i^{\text{out}} &= \Bigl\{\vec z_i \in X_i : \support(\vec z_i) \setminus \{e^1_i,e^2_i,e^3_i,e^4_i\} \neq \emptyset \Bigr\}
\end{align*}
contains a possibly empty subset of strategies that contains a non-critical resource $e \in E \setminus \{e^1_i,e^2_i,e^3_i,e^4_i\}$.
Next we describe how the set of strategies of both players are interweaved. For our purposes it is sufficient, to specify the critical resources of the two players. To this end, let $a,b,h,g$ be four resources such that
\begin{align*}
e_1^1 &= g, & e_2^1 &= g,\\
e_1^2 &= e, & e_2^2 &= h,\\
e_1^3 &= a, & e_2^3 &= a,\\
e_1^4 &= b, & e_2^4 &= b.
\end{align*}
Consider the following player-specific cost functions
\begin{align*}
c_{1,a}(x) &= \max \{0,x-1\}, & c_{2,a}(x) &= 1,\\
c_{1,b}(x) &= 1,              & c_{2,b}(x) &= 0,\\
c_{1,h}(x) &= 0,              & c_{2,h}(x) &= \max \{0,2x-2\},\\
c_{1,g}(x) &= \max \{0,3x-3\},& c_{2,g}(x) &= 2.
\intertext{For any non-critical resource $e \in E \setminus \{a,b,e,g\}$, we define}
c_{1,e}(x) &= 20,			 & c_{2,e}(x) &= 20.
\end{align*}

\begin{figure}[t]
\renewcommand{\arraystretch}{1.5}
{\scriptsize
\begin{tabular}{cc:cc:cccc:}
 strategy        & \multicolumn{1}{c}{} & $\vec y_2$   & \multicolumn{1}{c}{$\vec x_2$}   & \multicolumn{3}{c}{\raisebox{-0.35cm}{$\smash{\overbrace{\hspace{4cm}}^{\textstyle X_2^{\text{crit}}}}$}} & \multicolumn{1}{c}{$X_2^{\text{out}}$}\\
 \cline{2-2}
         & \multicolumn{1}{|c}{$\support$} & \multicolumn{1}{c}{$\{a,h\}$}	 & \multicolumn{1}{c}{$\{b,g\}$}		& $\{a,g\}$	& $\{h,g\}$ & $\{g\}$ \\
\cdashline{5-8}
\cline{3-4}
\multirow{1}{*}{$\vec y_1$} & \multicolumn{1}{c|}{$\{a,b\}$}
& $1\!+\!1,1\!+\!0$ & \multicolumn{1}{c|}{$0\!+\!1,0\!+\!2$} & $1\!+\!1,1\!+\!2$ & $0\!+\!1,0\!+\!2$ & $0\!+\!1,2\!+\!2$ & $~\cdot~~,\geq\!20$\\
\multirow{1}{*}{$\vec x_2$} & \multicolumn{1}{c|}{$\{h,g\}$}
& $0\!+\!0,1\!+\!2$ & \multicolumn{1}{c|}{$0\!+\!3,0\!+\!2$} & $0\!+\!3,1\!+\!2$ & $0\!+\!3,2\!+\!2$ & $0\!+\!6,2\!+\!2$ & $~\cdot~~,\geq\!20$\\
\cdashline{5-8}
\cline{3-4}
\multirow{3}{*}{$X_1^{\text{crit}} \left\{\rule{0cm}{0.6cm}\right.\hspace{-0.6cm}$} & $\{a,g\}$ 
& $1\!+\!0,1\!+\!0$ & $0\!+\!3,0\!+\!2$ & $1\!+\!3,1\!+\!2$ & $0\!+\!3,0\!+\!2$ & $0\!+\!6,2\!+\!2$ & $~\cdot~~,\geq\!20$\\
&$\{b,g\}$ 			& $1\!+\!0,1\!+\!0$ & $1\!+\!3,0\!+\!2$ & $1\!+\!3,1\!+\!2$ & $1\!+\!3,0\!+\!2$ & $1\!+\!6,2\!+\!2$ & $~\cdot~~,\geq\!20$\\
&$\{g\}$				& $3\!+\!3,1\!+\!0$ & $6\!+\!6,0\!+\!2$ & $6\!+\!6,1\!+\!2$ & $6\!+\!6,0\!+\!2$ & $9\!+\!9,2\!+\!2$ & $~\cdot~~,\geq\!20$\\
\multirow{1}{*}{$X_1^{\text{out}}\!\!\!\!\!$} & &
$\geq\!20,~\cdot~~$ & $\geq\!20,~\cdot~~$ & $\geq\!20,~\cdot~~$ & $\geq\!20,~\cdot~~$ & $\geq\!20,~\cdot~~$ & $\geq\!20,\geq\!20$ \\
\cdashline{3-8}
\end{tabular}
}
\caption{Game without a Nash equilibrium as constructed in the proof of Theorem~\ref{thm:no_nash}\label{fig:no_nash}}
\end{figure}

For a player~$i \in \{1,2\}$ and $e \in E$, we let $C_{i,e}(x_e;x_{-i,e}) := c_{i,e}(x_{i,e} + x_{-i,e})x_{i,e}$. The resulting private costs of the players are shown in Fig.~\ref{fig:no_nash}. Note that by Lemma~\ref{lem:non-submodular-strategies}, for every player~$i$ we are only guaranteed the existence of the two strategies $\vec x_i$ and $\vec y_i$ shown in the upper left part of the bimatrix. As shown in Lemma~\ref{lem:non-submodular-strategies}, any other strategy $
\vec z_i$ contains either a non-critical resource and is, thus, contained in the subset of strategies $X_i^{\text{out}}$ or contains only critical resources and is contained in the subset of strategies $X_i^{\text{crit}}$. The bimatrix in Fig.~\ref{fig:no_nash} has the property that there is no pure Nash equilibrium, no matter which subset of the strategies in $X_i^{\text{crit}}$ or whether a strategy in $X_i^{\text{out}}$ is actually present. We may thus conclude, that no matter how the sets of strategies described by $f$ specifically look like, no pure Nash equilibrium exists. 
\end{proof}

\section*{Acknowledgments}
We thank Satoru Fujishige for helpful
suggestions  improving the presentation of the paper.

\bibliographystyle{siam}
\bibliography{../../book/mrabbrev2012,../../book/journabbrev,../../book/congestion,../../book/confabbrev}

\begin{thebibliography}{10}

\bibitem{Ackermann09}
{\sc Heiner Ackermann}, {\em {N}ash Equilibria and Improvement Dynamics in
  Congestion Games}, PhD thesis, RWTH Aachen University, January 2009.

\bibitem{AckermannRV09}
{\sc Heiner Ackermann, Heiko R\"{o}glin, and Berthold V\"{o}cking}, {\em Pure
  {N}ash equilibria in player-specific and weighted congestion games}, Theoret.
  Comput. Sci., 410 (2009), pp.~1552--1563.

\bibitem{AnshelevichDKTWR08}
{\sc Elliot Anshelevich, Anirban Dasgupta, Jon Kleinberg, \'{E}va Tardos, Tom
  Wexler, and Tim Roughgarden}, {\em The price of stability for network design
  with fair cost allocation}, SIAM J. Comput., 38 (2008), pp.~1602--1623.

\bibitem{BALDICK1995}
{\sc Ross Baldick}, {\em Refined proximity and sensitivity results in linearly
  constrained convex separable integer programming}, Linear Algebra and Appl.,
  226 (1995), pp.~389 -- 407.

\bibitem{BeckmannMW56}
{\sc Martin Beckmann, C.~B. McGuire, and Christopher~B. Winsten}, {\em Studies
  in the Economics and Transportation}, Yale University Press, New Haven, CT,
  USA, 1956.

\bibitem{ChenR09}
{\sc Ho-Lin Chen and Tim Roughgarden}, {\em Network design with weighted
  players}, Theory Comput. Syst., 45 (2009), pp.~302--324.

\bibitem{CookGST86}
{\sc William~J. Cook, A.~M.~H. Gerards, Alexander Schrijver, and {\'{E}}va
  Tardos}, {\em Sensitivity theorems in integer linear programming}, Math.
  Program., 34 (1986), pp.~251--264.

\bibitem{DunkelS08}
{\sc Juliane Dunkel and Andreas~S. Schulz}, {\em On the complexity of
  pure-strategy {Nash} equilibria in congestion and local-effect games}, Math.
  Oper. Res., 33 (2008), pp.~851--868.

\bibitem{FedergruenG86b}
{\sc Awi Federgruen and Henri Groenevelt}, {\em The greedy procedure for
  resource allocation problems: Necessary and sufficient conditions for
  optimality}, Oper. Res., 34 (1986), pp.~909--918.

\bibitem{FotakisKS05}
{\sc Dimitris Fotakis, Spyros Kontogiannis, and Paul~G. Spirakis}, {\em Selfish
  unsplittable flows}, Theoret. Comput. Sci., 348 (2005), pp.~226--239.

\bibitem{fujishige2005submodular}
{\sc Satoru Fujishige}, {\em Submodular Functions and Optimization}, Elsevier,
  2005.

\bibitem{FujishigeGHP15}
{\sc Satoru Fujishige, Michel~X. Goemans, Tobias Harks, and Britta Peis}, {\em
  Matroids are immune to {Braess} paradox}.
\newblock Math. Oper. Res., to appear, 2016.

\bibitem{Gabow95}
{\sc Harold~N. Gabow}, {\em A matroid approach to finding edge connectivity and
  packing arborescences}, J. Comput. System Sci., 50 (1995), pp.~259--273.

\bibitem{GaiLK11}
{\sc Yi~Gai, Hua Liu, and Bhaskar Krishnamachari}, {\em A packet dropping
  mechanism for efficient operation of {$M/M/1$} queues with selfish users}, in
  Proc. 30th IEEE Int. Conf. Comput. Commun., 2011, pp.~2687--2695.

\bibitem{GairingMT11}
{\sc Martin Gairing, Burkhard Monien, and Karsten Tiemann}, {\em Routing
  (un-)splittable flow in games with player-specific linear latency functions},
  ACM Trans. Algorithms, 7 (2011), pp.~1--31.

\bibitem{groenevelt91}
{\sc Henri Groenevelt}, {\em Two algorithms for maximizing a separable concave
  function over a polymatroid feasible region}, European J. Oper. Res., 54
  (1991), pp.~227 -- 236.

\bibitem{HarksK12}
{\sc Tobias Harks and Max Klimm}, {\em On the existence of pure {Nash}
  equilibria in weighted congestion games}, Math. Oper. Res., 37 (2012),
  pp.~419--436.

\bibitem{HarksK16}
\leavevmode\vrule height 2pt depth -1.6pt width 23pt, {\em Congestion games
  with variable demands}, Math. Oper. Res., 41 (2016), pp.~255--277.

\bibitem{HarksKP14}
{\sc Tobias Harks, Max Klimm, and Britta Peis}, {\em Resource competition on
  integral polymatroids}, in Proc. 10th Int. Conf. Web and Internet Econ.,
  Tie-Yan Liu, Qi~Qi, and Yinyu Ye, eds., vol.~8877 of LNCS, 2014,
  pp.~189--202.

\bibitem{HaurieM85}
{\sc Alain Haurie and Patrice Marcotte}, {\em On the relationship between
  {N}ash-{C}ournot and {W}ardrop equilibria}, Networks, 15 (1985),
  pp.~295--308.

\bibitem{HeZZ12}
{\sc Simai He, Jiawei Zhang, and Shuzhong Zhang}, {\em Polymatroid
  optimization, submodularity, and joint replenishment games}, Oper. Res., 60
  (2012), pp.~128--137.

\bibitem{HochbaumS90}
{\sc Dorit~S. Hochbaum and J.~George Shanthikumar}, {\em Convex separable
  optimization is not much harder than linear optimization}, J. {ACM}, 37
  (1990), pp.~843--862.

\bibitem{IeongMNSS05}
{\sc Samuel Ieong, Robert McGrew, Eugene Nudelman, Yoav Shoham, and Qixiang
  Sun}, {\em Fast and compact: A simple class of congestion games}, in Proc.
  20th Nat. Conf. Artificial Intell., 2005, pp.~489--494.

\bibitem{JohariT06}
{\sc Ramesh Johari and John~N. Tsitsiklis}, {\em A scalable network resource
  allocation mechanism with bounded efficiency loss.}, IEEE J. Sel. Areas
  Commun., 24 (2006), pp.~992--999.

\bibitem{KellyMT98}
{\sc Frank~P. Kelly, Aman~K. Maulloo, and David K.~H. Tan}, {\em Rate control
  in communication networks: Shadow prices, proportional fairness, and
  stability}, J. Oper. Res. Soc., 49 (1998), pp.~237--252.

\bibitem{KorilisL95}
{\sc Yannis~A. Korillis and Aurel~A. Lazar}, {\em On the existence of
  equilibria in noncooperative optimal flow control}, J. ACM, 42 (1995),
  pp.~584--613.

\bibitem{KrystaSV03}
{\sc Piotr Krysta, Peter Sanders, and Berthold V{\"o}cking}, {\em Scheduling
  and traffic allocation for tasks with bounded splittability}, in Proc. 28th
  Int. Symp. Math. Found. Comput. Sci., Branislav Rovan and Peter Vojtas, eds.,
  vol.~2747 of LNCS, 2003, pp.~500--510.

\bibitem{Meyers08}
{\sc Carol Meyers}, {\em Network Flow Problems and Congestion Games: Complexity
  and Approximation Results}, PhD thesis, MIT, Operations Research Center,
  2006.

\bibitem{Milchtaich96}
{\sc Igal Milchtaich}, {\em Congestion games with player-specific payoff
  functions}, Games Econom. Behav., 13 (1996), pp.~111--124.

\bibitem{Milchtaich06}
\leavevmode\vrule height 2pt depth -1.6pt width 23pt, {\em The equilibrium
  existence problem in finite network congestion games}, in Proc. 2nd Int.
  Workshop Internet and Network Econ., Marios Mavronicolas and Spyros
  Kontogiannis, eds., vol.~4286 of LNCS, 2006, pp.~87--98.

\bibitem{MoriguchiST11}
{\sc Satoko Moriguchi, Akiyoshi Shioura, and Nobuyuki Tsuchimura}, {\em
  M-convex function minimization by continuous relaxation approach: Proximity
  theorem and algorithm}, SIAM J. Optim., 21 (2011), pp.~633--668.

\bibitem{Murota:2003}
{\sc Kazuo Murota}, {\em Discrete Convex Analysis: Monographs on Discrete
  Mathematics and Applications 10}, Society for Industrial and Applied
  Mathematics, Philadelphia, PA, USA, 2003.

\bibitem{MurotaT04}
{\sc Kazuo Murota and Akihisa Tamura}, {\em Proximity theorems of discrete
  convex functions}, Math. Program., 99 (2004), pp.~539--562.

\bibitem{Rosenthal73congestion}
{\sc Robert~W. Rosenthal}, {\em A class of games possessing pure-strategy
  {N}ash equilibria}, Internat. J. Game Theory, 2 (1973), pp.~65--67.

\bibitem{Rosenthal73integers}
\leavevmode\vrule height 2pt depth -1.6pt width 23pt, {\em The network
  equilibrium problem in integers}, Networks, 3 (1973), pp.~53--59.

\bibitem{Roughgarden05book}
{\sc Tim Roughgarden}, {\em Selfish Routing and the Price of Anarchy}, MIT
  Press, Cambridge, MA, USA, 2005.

\bibitem{Srikant03}
{\sc Rayadurgam Srikant}, {\em The Mathematics of {Internet} Congestion
  Control}, Birkh{\"a}user, Basel, Switzerland, 2003.

\bibitem{Topkis78}
{\sc Donald~M Topkis}, {\em Minimizing a submodular function on a lattice},
  Oper. Res., 26 (1978), pp.~305--321.

\bibitem{Topkis98}
{\sc Donald~M. Topkis}, {\em Supermodularity and Complementarity}, Princeton
  University Press, Princeton, NJ, USA, 1998.

\bibitem{TranPCRJ11}
{\sc Long Tran-Thanh, Maria Polukarov, Archie Chapman, Alex Rogers, and
  Nicholas~R. Jennings}, {\em On the existence of pure strategy {Nash}
  equilibria in integer-splittable weighted congestion games}, in Proc. 4th
  Int. Symp. Algorithmic Game Theory, G.~Persiano, ed., vol.~6982 of LNCS,
  2011, pp.~236--253.

\bibitem{Wardrop52}
{\sc John~G. Wardrop}, {\em Some theoretical aspects of road traffic research},
  Proc. Inst. Civil Engineers, 1 (1952), pp.~325--362.

\bibitem{Yao02}
{\sc David~D. Yao}, {\em Dynamic scheduling via polymatroid optimization}, in
  Performance Evaluation of Complex Systems: Techniques and Tools, Performance
  2002, Tutorial Lectures, 2002, pp.~89--113.

\end{thebibliography}

\end{document}